\documentclass[a4paper,11pt]{article}

\usepackage{fullpage}
\usepackage{authblk}

\usepackage[utf8]{inputenc} 
\usepackage[T1]{fontenc}    
\usepackage[hidelinks]{hyperref}
\hypersetup{
    colorlinks,
    linkcolor={blue},
    citecolor={blue},
    urlcolor={blue}
}
\usepackage{url}            
\usepackage{booktabs}       
\usepackage{amsfonts}       
\usepackage{nicefrac}       
\usepackage{microtype}      
\usepackage[table,dvipsnames]{xcolor}
\usepackage{comment}

\usepackage{enumitem}

\usepackage{subfigure}

\usepackage{setspace}
\usepackage[linesnumbered,ruled,vlined]{algorithm2e}

\SetCommentSty{mycommfont}

\usepackage{amssymb,amsmath,amsfonts,amstext,amsthm, mathtools}
\usepackage{cleveref}
\Crefname{equation}{Eq.}{Eqs.}

\usepackage[font=small]{caption}

\usepackage{natbib}

\renewcommand{\cite}[1]{\citep{#1}}

\usepackage{pgfplots}
\usepackage{pgfplotstable}

\usepackage{tikz}
\usetikzlibrary{shapes.geometric, shapes.misc, automata, positioning, arrows, arrows.meta, decorations.markings}
\usetikzlibrary{calc}
\usepackage{pgfmath}

\usepackage{thm-restate}

\usepackage{boldline,multirow}

\usepackage{ifthen}
\newboolean{isrestating}
\setboolean{isrestating}{false}
\newcommand{\hideinrestate}[1]{\ifthenelse{\boolean{isrestating}}{}{#1}}

\newtheorem{thm}{_} 

\newtheorem{corollary}[thm]{Corollary}

\newtheorem{claim}{Claim}

\theoremstyle{definition}
\newtheorem{definition}{Definition}

\newenvironment{talign*}
 {\csname align*\endcsname}
 {\endalign}

\newcommand{\poly}{\operatorname{poly}}
\newcommand{\init}{\mathrm{init}}

\newcommand{\rmax}{r_{\max}}
\newcommand{\rmin}{r_{\min}}
\newcommand{\ain}{\mathtt{in}} 
\newcommand{\aout}{\mathtt{out}} 
\newcommand{\vv}{\mathbf{v}}
\newcommand{\xcalV}{\widehat{\calV}}
\newcommand{\xPhi}{\widehat{\Phi}}
\newcommand{\tPhi}{\xPhi^{\mathrm{TB}}}
\newcommand{\conv}{\operatorname{conv}}
\newcommand{\neigh}{\operatorname{neigh}}
\newcommand{\bellman}{\operatorname{Bellman}}
\newcommand{\Gsub}{\calG_\textnormal{subs}}
\newcommand{\Gtrun}{\calG_\textnormal{trun}}

\renewcommand{\emptyset}{\varnothing}

\usepackage{bm}

\newcommand{\va}{\mathbf{a}}
\newcommand{\bb}{\mathbf{b}}

\newcommand{\vg}{\mathbf{g}}

\newcommand{\rr}{\mathbf{r}}
\newcommand{\uu}{\mathbf{u}}
\newcommand{\ww}{\mathbf{w}}
\newcommand{\xx}{\mathbf{x}}
\newcommand{\yy}{\mathbf{y}}
\newcommand{\zz}{\mathbf{z}}

\newcommand{\bfA}{\mathbf{A}}

\newcommand{\calB}{\mathcal{B}}
\newcommand{\calC}{\mathcal{C}}

\newcommand{\calF}{\mathcal{F}}
\newcommand{\calG}{\mathcal{G}}

\newcommand{\calV}{\mathcal{V}}

\DeclareMathOperator*{\argmax}{arg\,max}

\DeclareMathOperator{\E}{\mathbb{E}}
\newcommand{\given}{{\,|\,}}

\title{{\bf Value-Set Iteration}: Computing Optimal Correlated Equilibria in Infinite-Horizon Multi-Player Stochastic Games}

\author[1]{Jiarui Gan}
\author[2]{Rupak Majumdar}

\affil[1]{University of Oxford}
\affil[2]{Max Planck Institute for Software Systems (MPI-SWS)}

\date{}

\begin{document}

\maketitle

\begin{abstract}
We study the problem of computing optimal correlated equilibria (CEs) in infinite-horizon multi-player stochastic games, where correlation signals are provided over time.
In this setting, optimal CEs require history-dependent policies; this poses new representational and algorithmic challenges as the number of possible histories grows exponentially with the number of time steps.
We focus on computing $(\epsilon, \delta)$-optimal CEs---solutions that achieve a value within $\epsilon$ of an optimal CE, while allowing the agents' incentive constraints to be violated by at most $\delta$.
Our main result is an algorithm that computes an $(\epsilon,\delta)$-optimal CE in time polynomial in $1/(\epsilon\delta(1 - \gamma))^{n+1}$, where $\gamma$ is the discount factor, and $n$ is the number of agents.
For (a slightly more general variant of) turn-based games, we further reduce the complexity to a polynomial in $n$. 
We also establish that the bi-criterion approximation is necessary by proving matching inapproximability bounds.

Our technical core is a novel approach based on \emph{inducible value sets}, which leverages a compact representation of history-dependent CEs through the values they induce to overcome the representational challenge. 
We develop the \emph{value-set iteration} algorithm---which operates by iteratively updating estimates of inducible value sets---and characterize CEs as the greatest fixed point of the update map. Our algorithm provides a groundwork for computing optimal CEs in general multi-player stochastic settings.
\end{abstract}

\section{Introduction}

\begin{quote}
{\em
A finite game is played for the purpose of winning, an infinite game for the purpose of continuing the play.}

\hfill--- James P. Carse, Finite and Infinite Games 
\end{quote}

\noindent
Stochastic games (or Markov games), introduced by \citet{shapley1953stochastic}, are dynamic games that evolve in time under probabilistic state transitions. In each time step, players' actions determine both the immediate payoffs and the next state of the game. This framework generalizes Markov decision processes (MDPs) to the multi-player setting, capturing the interplay of strategic behavior and uncertainty.

In this paper, we study the problem of computing {\em optimal correlated equilibria} (CEs) in $n$-player, general-sum, stochastic games.
Without loss of generality, we view the game as one played between a coordinator---referred to as the {\em principal}---and a set of {\em agents}.
At each time step, the principal selects a joint action and recommends each agent to play the corresponding action through a private communication channel.
The agents receive the recommendations and, simultaneously and independently, each decides an action to play, possibly one different from the recommendation.
Jointly, the agents' actions yield a reward for every player (including the principal), and the actions result in the environment transitioning to a new state. 
When the agents are all incentivized to play the recommended actions throughout the game, the distributions from which the principal draws the joint actions form a correlated equilibrium (CE) (among the agents). 
Our objective is to compute an \emph{optimal} CE, one that maximizes the principal's value, or cumulative reward, among all CEs.

More precisely, we consider the {\em extensive-form} CE (EFCE) \cite{von2008extensive} in this paper, where recommendation signals are provided step by step. 
Each signal recommends only the action for the current time step. This CE concept is different from the {\em normal-form} CE, where recommendation signals are provided all at once at the beginning of the game, each indicating an entire sequence of actions for a player to perform, from the beginning to the end of the game. 
The EFCE is a more suitable solution concept for scenarios where decision-making is interleaved with information exchange, which is a more common practice.
For simplicity, we will refer to EFCEs as CEs throughout. 


\colorlet{mygray}{gray!90!blue}
\definecolor{mywhite}{RGB}{255,255,255}
\colorlet{myblue}{blue!90!cyan}

\tikzset{
->, 
node distance=1.7cm, 
every state/.style={very thick, fill=mywhite, minimum size=7mm},
every edge/.append style={thick},
initial text= ,
}

\tikzstyle{labelOnEdge}=[fill=white,minimum height=3mm,minimum width=4mm,inner sep=1mm,text width=3mm]
\tikzstyle{myTermState}=[state, thick, accepting, minimum size=4mm,text width=0mm]
\tikzstyle{principal}=[fill=mydarkgray, text=white, draw=black, font=\normalsize] 


\begin{figure}
\centering
\subfigure[~]{
\renewcommand{\arraystretch}{1.5}
\begin{tikzpicture}[baseline={(current bounding box.center)}, node distance=1.6cm, scale=1, transform shape]

\node[state, initial, initial where=above, inner sep=0pt, text width=6mm, align=center] (s) {$s$};
\node[myTermState, below left of=s, xshift=2mm] (t) {};

\draw 
(s) edge[] node[left, anchor=east, yshift=3mm]{} (t)
(s) edge[out=-60, in=-25, looseness=11] node[right]{$(C,C)$} (s)
;
\end{tikzpicture}
\qquad
\small
\begin{tabular}{ccc}
                         & $C$                      & $D$                      \\ \cline{2-3} 
\multicolumn{1}{l|}{$C$} & \multicolumn{1}{c|}{$~1,\,1~$} & \multicolumn{1}{c|}{$-1,\, 2$} \\ \cline{2-3} 
\multicolumn{1}{l|}{$D$} & \multicolumn{1}{c|}{$2,\,-1$} & \multicolumn{1}{c|}{$~0,\,0~$} \\ \cline{2-3} 
\end{tabular}
}

\subfigure[~]{
\begin{tikzpicture}[baseline={(current bounding box.center)}, node distance=1.6cm, scale=1, transform shape]
\tikzstyle{every node}=[font=\small] 

\node[state, initial, initial where=above] (s1) {$s_1$};
\node[myTermState, below left of=s, xshift=3mm] (t1) {};
\node[state, right of=s1, xshift=10mm] (s2) {$s_2$};
\node[myTermState, below left of=s2, xshift=3mm] (t2) {};
\node[state, right of=s2, xshift=18mm] (sk) {$s_t$};
\node[myTermState, right of=sk, , xshift=8mm] (tt) {};
\node[myTermState, below left of=sk, xshift=3mm] (tk) {};

\draw 
(s1) edge[] node[above]{$(C,C)$} (s2)
(s2) edge[] node[labelOnEdge, text width=8mm, align=center] {\Large $\dots$} (sk)
(s1) edge[] (t1)
(s2) edge[] (t2)
(sk) edge[] (tk)
(sk) edge[] node[above]{$(C,C)$} (tt)
;
\end{tikzpicture}
\vspace{3mm}
}
\vspace{4mm}
\caption{(a) A game that continues if both players cooperate ($C$) and terminates if some of them defects ($D$). The immediate rewards are given in the matrix on the right.
(b) A finite-horizon approximation of the game in (a) involving $t$ time steps. 
If both players adopt the discount factor $\gamma = 2/3$, then playing $(C,C)$ throughout can be sustained as an equilibrium in (a). 
In contrast, any finite-horizon approximation of this game cannot sustain $(C,C)$ as a correlated equilibrium.
This can be seen by induction: In the last time step, $D$ strictly dominates $C$, so $(D,D)$ is the only possible equilibrium; in turn, given that $(D,D)$ will be played in step $t$, the same conclusion can be drawn for step $t-1$. As a result, only $(D,D)$ can be sustained as an equilibrium in (b), no matter how large $t$ is. 
}
\label{fig:intro}
\end{figure}
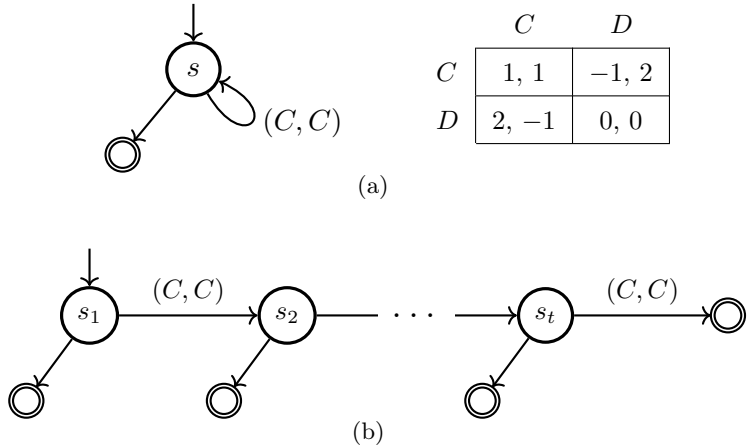

Another fundamental aspect of our work is the consideration of \emph{infinite-horizon} games, which naturally subsume their finite-horizon counterparts. 
The benefit of considering an infinite-horizon model over a finite one may seem marginal in many problem settings. 
Under reward discounting in particular, the discrepancy between the two diminishes at an exponential rate as the time horizon increases. 
This perspective has made finite-horizon models a convenient and widely adopted choice across various fields. (In reinforcement learning, for example, episodic settings---where each episode consists of a repeated finite-horizon game---are fairly common.)
While this has worked well in single-agent or fully cooperative settings, finite-horizon models may fail to serve as faithful approximations to their infinite-horizon counterparts in multi-player settings. 
When incentives are misaligned and players may deviate at any time step, even a minor shift in the distant future may trigger a ``chain reaction'' that propagates all the way backward to the present and eventually alters the outcome substantially.
\Cref{fig:intro} provides a concrete example to illustrate such issues, where the equilibrium payoffs in the original infinite-horizon game and its finite-horizon approximation differ by a constant gap, no matter how many time steps the latter spans.

As manifested through the example, the finite-horizon approximation fails due to a lack of long-term vision beyond the last time step.
We can remedy this by incorporating into the finite-horizon model an estimate of the continuation values---values that could have been yielded had the process continued. 
Indeed, accurately estimating these values is key to the computational approach we propose in this paper. 
Underestimating the values can lead to an overly narrow view of the equilibrium space and, consequently, suboptimality of the equilibria computed. 
On the other hand, overestimation risks creating ``bubbles'' that cannot be sustained if the game continues.

At a high level, this approach is similar to many other value-based methods, which operate by evaluating values of the states. 
The difference is that, while classical value iteration estimates a single value for each state, we need to maintain a {\em set} of values for each state as we no longer focus on stationary policies.
As we will demonstrate, stationary policies---which assign a fixed distribution of (joint) actions to each state---is no longer without loss of optimality in our setting (unlike in MDPs or two-player zero-sum stochastic games); worse still, optimal stationary policies are computationally intractable.
Surprisingly though, both issues can be addressed by employing history-dependent policies. 
The main challenge is to find an effective representation of history-dependent policies to enable the computation as the number of possible histories grows exponentially with length of the time horizon.
The {\em value-set iteration} algorithm, which we develop as the main result of this paper, tackles this challenge.

\subsection{Our Results}

Key to our approach is the shift from stationary to history-dependent policies/CEs.\footnote{There is yet another type of {\em non}-stationary policies---often called {\em Markovian policies}---which, while still mapping the current state to an action distribution, allow this mapping to vary over time (i.e., they remember but only the time step). Similarly to stationary policies, Markovian policies are suboptimal and computationally intractable in our setting (see \Cref{thm:inapproximability}).\label{fn:markovian}}  
History-dependent policies extend stationary ones by considering the entire history of past play (including the current state), mapping each possible history to an action distribution.
While this extension makes history-dependent policies more powerful, it comes with a representational challenge: the number of possible histories grows exponentially with the horizon length, making explicit representations intractable even in the finite-horizon setting.
Given this barrier, our approach builds on a compact representation of history-dependent policies, which encodes each policy as the values they induce.
Additionally, we maintain for each state their {\em inducible value set}, which contains all the inducible values at that state.
With these inducible value sets, we can {\em unroll} any compactly represented history-dependent policy---that is, compute the distribution it assigns to any given history. 
Effectively, this allows us to execute a policy while computing the distributions required on the fly, even though it is infeasible to pre-compute and explicitly write down the entire policy.
Following this approach, the problem reduces to computing the inducible value sets. We introduce a {\em value-set iteration} algorithm to accomplish this task.

\paragraph{Value-Set Iteration}
Much like the {\em value iteration} method for solving MDPs, our value-set iteration algorithm iteratively refines an estimate of each inducible value set. 
However, the underlying dynamics are fundamentally different. 
Standard value iteration converges due to a {\em contraction} property of the Bellman operator, which ensures a unique fixed point via the Banach fixed-point theorem. 
Our update operator, however, lacks this property. To establish the convergence, we rely on a {\em monotonicity} property, which makes the inducible value sets a fixed point following Tarski's fixed-point theorem. 
Tarski's theorem asserts that the set of fixed points of such a monotonic operator forms a complete lattice, meaning in particular that it has a greatest fixed point. 
We show that the inducible value set is exactly the greatest fixed point and is always nonempty; moreover, it can be attained at the limit of value-set iteration, as the number of iterations approaches infinity.

\paragraph{Approximating Optimal CEs and Matching Inapproximability Results}
In practice, we terminate value-set iteration after a finite number of iterations, whereby the algorithm computes an {\em approximately fixed} point. 
We show that any vector in this approximately fixed point can be induced, up to an $\epsilon$ error, by a $\delta$-CE---one that may violate the agents' incentive constraints but by no more than a small $\delta$. 
As a result, the best achievable value within the approximately fixed point gives an $(\epsilon,\delta)$-optimal CE: its value is at most $\epsilon$ worse than the best {\em exact} CE, which satisfies all incentive constraints exactly without any violation. 
This approximation approach is the same as {\em resource augmentation} in approximation algorithm design, where we allow the algorithm to use $\delta$ more resources while still measuring its performance against the best possible solution with the original (unaugmented) resources.
We argue that such approximate solutions are the best one can hope for, given two inapproximability results we establish.
\begin{itemize}
\item 
First, representing (exact) optimal policies may require irrational numbers, and the gap between optimal policies and the best policy involving only rational numbers can be arbitrarily large.

\item
Second, without resource augmentation, even to find a constant-factor approximation to the optimal solution is NP-hard.
\end{itemize}

For games with $n$ agents, our approximation algorithm computes an $(\epsilon,\delta)$-optimal CE in time polynomial in the size of the game instance and $(\epsilon \delta (1-\gamma))^{-(n+1)}$, for any desired accuracy parameters $\epsilon$ and $\delta$, and discount factor $\gamma$.
Hence, for constant $\epsilon$, $\delta$, and $\gamma$, the time complexity is polynomial in the input size if the game is given in the matrix form---whereby rewards and transition probabilities are enumerated explicitly for every joint action, so the size of the representation is already exponential in $n$.
Indeed, due to the exponential growth of the joint action space, succinct representation is typically preferred in multi-player games, where rewards and transition probabilities are specified through efficiently computable functions.
We thus investigate whether the exponential dependence on $n$ in our approach can be overcome in some classes of succinctly represented games.

\paragraph{Faster Algorithm for Turn-Based Games}
We show that in a slightly more general class of {\em turn-based games} that allows a constant number of players to act simultaneously, the time complexity of value-set iteration can be reduced to a polynomial in $n$.
We present algorithms based on the concept of {\em $\lambda$-memory meta-game}, where each meta-state tracks not only the current state of the original game but also states in the previous $\lambda$ steps for a suitably chosen constant $\lambda$. 
This transformation allows us to effectively approximate the inducible value set while reducing the dimension of the value space from $n$ to a constant.

\medskip

We remark that while our primary goal is computing an {\em optimal} CE, our algorithms readily solve the problem of computing one (not necessarily optimal) CE, which is of broad interest in the literature, too.
It remains open though whether it is possible to bound the time complexity of computing one CE within a polynomial in $n$, which has been shown possible in the one-shot setting for a wider range of succinctly represented games.
We provide a discussion about this open question in  \Cref{sec:conclusion}.

\subsection{Related Work}

Stochastic games have been studied extensively for decades, given their broad impact in fields such as control theory, reinforcement learning, and economics. 
There is a plethora of work exploring various subclasses of stochastic games under different problem settings and solution concepts.

\paragraph{Two-Player Zero-Sum Games}

In \citet{shapley1953stochastic}'s seminal paper, where he introduced stochastic games, he considered two-player zero-sum games. In these games, the set of CEs coincides with both the set of Nash equilibria and the set of minimax equilibria. Shapley proposed a value iteration procedure that iteratively updates the minimax value of the stage game. 
This procedure converges exponentially fast and produces a pair of stationary equilibrium strategies.
Nonetheless, exact equilibria may include irrational numbers, an issue we also examine in this paper. 
Solving the game exactly requires algebraic number computations and exponential-time algorithms have been developed \cite{hansen2011exact,oliu2021new}.
As a practical approach, we focus on near-optimal solutions instead. In more specific forms of two-player zero-sum games---those that are also turn-based in particular---an equilibrium of deterministic stationary strategies always exists and involves only rational numbers. For such games, \citet{hansen2013strategy} provided a strongly polynomial-time algorithm.

\paragraph{General-Sum Games}

Many of the appealing properties of two-player zero-sum games do not extend to multi-player general-sum settings. 
In the latter, stationary CEs no longer guarantee optimality. For example, ``tit-for-tat'' strategies require remembering players' actions in the past and cannot be implemented by using stationary strategies; such tactics are often necessary for achieving optimality in general-sum games.
Despite this limitation, stationary strategies remain a popular choice in the literature due to their simplicity. Computing a stationary Nash equilibrium is PPAD-complete, as it is even in one-shot games \cite{daskalakis2009complexity} and even with two players \cite{chen2009settling}. 
While CEs are tractable in one-shot games, \citet{daskalakis2023complexity} showed that, in stochastic games, computing a stationary CE is PPAD-hard, even under the weaker notion of \emph{coarse} CE. Moreover, computing optimal stationary (or constant-memory) CEs is NP-hard \cite{letchford2012computing}.
In sequential persuasion games (a slightly richer form of games involving private observations of the coordinator), \citet{gan2022bayesian} showed that an optimal CE-like solution, when restricted to a stationary one, is inapproximable. In this paper, we prove a similar inapproximability result under an even more relaxed approximation notion.

\paragraph{Finite-Horizon Games}

In finite-horizon games, once history-dependent strategies are allowed, computing \emph{one} CE is straightforward via backward induction: at each time step, select a CE of the stage game, based on the immediate rewards as well as the continuation values of the CEs selected for the subsequent time steps. 
Hence, a more interesting problem is to compute an \emph{optimal} CE. 
There has been recent work on computing an optimal CE of two-player general-sum \emph{turn-based} games \cite{zhang2023efficiently}, where an exact optimal solution has been shown tractable by way of querying the Pareto frontiers of the value set.
Similar CE-like solution concepts have also been studied in finite-horizon persuasion games \citet{gan2023sequentialv2,bernasconi2024persuading}, where near-optimal history-dependent solutions are shown to be tractable.

Extensive-form games (EFGs) are another common finite-horizon model, often featuring a game tree of limited depth. Compared to our stochastic game model, EFGs are easier in that the number of histories is bounded by the size of the game tree, so it does not grow exponentially with the size of the problem instance. But they can be more complicated as private information in EFGs may be invisible to a player across multiple steps. (In our model the entire interaction history becomes common knowledge at the end of each time step.)
Indeed, with such higher degrees of information asymmetry (which is not our focus), even in two-player EFGs computing an optimal CE (whether extensive-form or normal-form) has been shown to be NP-hard, although it becomes tractable if the game does not contain any chance nodes \cite{von2008extensive}. Meanwhile, finding \emph{one} EFCE remains tractable \cite{huang2008computing}, via methods similar to those for solving succinctly represented one-shot games \cite{papadimitriou2008computing}.
There have also been other recent works examining the computation of EFCE in various more specific types of EFGs \cite{farina2020polynomial,zhang2022optimal,zhang2022polynomial}.

\paragraph{Value-Set Iteration}
\citet{murray2007finding} introduced an algorithm for computing CEs, using a fixed point characterization similarly to ours.
However, they did not consider the approximation framework and relaxation necessary for finite-time convergence.
Later extensions by \citet{dermed2009solving} and \citet{macdermed2011quick} incorporated approximations to ensure finite-time convergence, 
but they lack formal analysis to substantiate the key claims, leaving the soundness and optimality of their methods uncertain.
\citet{kitti2016subgame} reintroduced the same iteration method as that by \citet{murray2007finding}, but did not address any algorithmic problems. 
Our approach resonates with these previous ideas while it provides a complete algorithmic analysis, including proofs of convergence and correctness, as well as a systematic study of suitable approximation criteria backed by matching inapproximability results. 
In the finite-horizon settings, value-set iteration simplifies to a dynamic programming approach that builds value sets backward from the last time step. This idea, along with related approaches that track value sets or Pareto frontiers, has appeared in several studies of finite-horizon models \cite{letchford2010computing,bovsansky2017computation,gan2023sequentialv2,zhang2023efficiently}.
However, new ideas were needed to extend to an infinite horizon.

\paragraph{Tarski Fixed Point}

Similarly to other value iteration methods, our value-set iteration approach converges to a fixed point of the update operator. It is therefore naturally associated with fixed-point theorems.
In our case, Tarski's fixed-point theorem \cite{tarski1955lattice} applies due to the monotonicity of the update operator. 
Recent work has investigated the computational aspects of this theorem and developed algorithms for computing Tarski fixed points \cite{etessami2019tarski,chen2022improved,dang2024computations}. 
These general algorithms can, in principle, be applied directly to computing equilibria in stochastic games.
In two-player zero-sum games, this yields exponential-time algorithms (in the number $m$ of states) and gives the best-known time complexity upper bound when the discount factor is also a variable in the problem input (so it can be exponentially close to $1$, rather than being a constant as in our case). 
A more recent improvement uses both contraction (which holds in the two-player zero-sum setting) and monotonicity to reduce the exponent \cite{batziou2024monotone}, though still leaving an exponential dependence in $m$.

As for the general-sum setting we study, several challenges arise. 
First, while the operator for two-player zero-sum games maps a vector encoding the players' values in $\mathbb{R}^n$ to another such vector (where $n$ is the number of players), in our case, the operator maps a {\em set} of vectors in $\mathbb{R}^n$ to another set.
To represent these sets, a common approach is to use convex polytopes, but this leads to a representation with dimensionality exponential in $n$.
Hence, direct applications of general algorithms for computing Tarski fixed-points result in {\em doubly exponential} time complexity in $n$ (and exponential in $m$).
In practice, this is less desirable than an algorithm with exponential dependence on $n$ under a constant discount factor.
Second, we need not just an arbitrary fixed point but a greatest one since our goal is to compute an optimal CE. This issue does not arise in two-player zero-sum games---where a unique fixed point is ensured due to contraction---but it becomes crucial in our multi-player general-sum setting and further complicates the problem.

\section{Preliminaries}

We consider a fairly standard model of stochastic games.   
Let there be $n+1$ players, including a principal (player $0$) and a set $N$ of $n$ agents (players $1,\dots,n$). The game is given by a tuple $\calG = \langle S, s_\init, A, p, \rr, \gamma \rangle$ consisting of:
a finite state space $S$, 
an initial state $s_\init \in S$, 
a finite action space $A$ for each player,
a discount factor $\gamma \in [0,1)$, 
a transition function $p$, and a set of reward functions $\rr = (r_0, r_1, \dots, r_n)$.
The transition function $p : S \times \bfA \to \Delta(S)$ defines the probability $p(s' \given s, \va)$ of the state transitioning from $s$ to $s'$ after a joint action $\va \in \bfA \coloneqq A^{n+1}$ is performed.
Each reward function $r_i: S \times A \to \mathbb{R}$ defines the reward $r_i(s,\va)$ for player~$i$ when a joint action $\va$ is performed at state $s$.

The game is played over an {\em infinite} horizon (which generalizes finite-horizon models). 
All players are far-sighted and aim to maximize the discounted sums of their individual rewards (under the $\gamma$ discount factor).
The principal plays the role of the coordinator in the game.
W.l.o.g., we allow the principal to perform actions too, just as the agents---in scenarios where this is not possible, it suffices to make the reward and transition functions invariant w.r.t. the principal's action.

The interaction within each time step $t$ is as follows:
\begin{itemize}
\item 
The principal selects a joint action $\va = (a_0, \dots, a_n) \in \bfA$ and recommends each agent $i$ to perform $a_i$.
Each recommendation $a_i$ is sent through a private channel to agent $i$, unobservable to other agents.

\item 
Then, simultaneously, the principal performs $a_0$, and every agent plays an action they prefer. 

\item 
The players observe the next state resulting from their joint action, as well as the actions performed by each other. 
\end{itemize}
The game has perfect recall: every player remembers the entire history.

\subsection{Correlation Policy and Correlated Equilibrium}

The principal commits to a correlation policy, which specifies how the joint actions are selected throughout the game. 
In the most generic form, a policy is history-dependent: it is a function $\pi: \Sigma \times S \to \Delta(\bfA)$,
where $\Sigma = \{\varnothing\} \cup \bigcup_{t=1}^{+\infty} (S \times \bfA^2)^t$ contains all possible sequences of the interaction history, with $\varnothing$ representing a special empty sequence.
Each sequence $\sigma = (s^1,\va^1, \bb^1; s^2,\va^2, \bb^2;\dots;s^\ell, \va^\ell, \bb^\ell) \in \Sigma$ records, for each time step $t \in \{1,\dots,\ell\}$, the state $s^t$, the joint action $\va^t$ recommended, and the joint action $\bb^t$ actually performed by the players.\footnote{We use semicolons in the notation to separate elements that belong to different time steps.} 
We denote by $|\sigma|= \ell$ the number of time steps involved in $\sigma$. 

Since $\pi(\sigma; s)$ is a distribution over $\bfA$, we denote by $\pi( \va \given \sigma; s)$ the probability of each $\va \in \bfA$ in this distribution.
A {\em stationary} policy $\pi$ is a special type of history-dependent policy where $\pi(\sigma;s) = \pi(\sigma';s)$ holds for all $\sigma,\sigma' \in \Sigma$ and $s \in S$.

\paragraph{Correlated Equilibrium}

The agents are not obliged to play actions recommended by the principal---they can only be incentivized to do so if the recommended actions are optimal with respect to their own objectives. 
Indeed, as a rule of thumb, it is without loss of optimality to consider {\em incentive compatible} (IC) policies, which always incentivize the agents to adhere to the recommended actions. 
The action distributions in an IC policy form precisely a CE.

More formally, consider a {\em deviation plan} $\rho: \Sigma \times S \times A \to A$ of an agent. 
When $\rho$ is adopted, the agent, upon being recommended an action $a$, plays the action $\rho(\sigma;s,a)$, based on the history $\sigma$ and the current state $s$.
For a joint action $\va \in \bfA$, we denote by $\va \oplus_i b$ the new joint action resulting from a unilateral deviation of agent $i$ to an action $b \in A$. 

The V- and Q-values induced by $\pi$ and a deviation plan $\rho$ of agent~$i$ (assuming all other players play according to $\pi$) are defined in order:
\begin{align}
\label{eq:V}
V^{\pi,\rho} \left(\sigma; s \right) 
&\coloneqq
\E_{\va \sim \pi(\cdot \given \sigma, s)} Q^{\pi,\rho} (\sigma; s, \va)
\\
\label{eq:Q}
Q^{\pi,\rho} (\sigma; s, \va) 
&\coloneqq
\rr \left(s,\, \va' \right) + 
\E_{s'\sim p\left(\cdot \given s,\, \va' \right)} \gamma  V^{\pi,\rho} \left(\sigma; s, \va'; s' \right), 
\end{align}
where $\va' = \va \oplus \rho(\sigma; s, a_i)$.
Note that both functions output vectors, and we denote by $Q_i^{\pi,\rho}$ and $V_i^{\pi,\rho}$ the $i$-th components of the vectors, which correspond to player $i$.
Following these definitions, the V-value captures the players' expected cumulative rewards from time step $|\sigma|+1$ onward, conditioned on the sequence $(\sigma; s)$, i.e.,
\[
V^{\pi,\rho} (\sigma; s ) = \E \Big( \sum_{t=|\sigma|+1}^\infty \gamma^{t - |\sigma| - 1}  \cdot \rr(s^t, \va^t) \,\Big|\, \sigma; s \Big),
\]
where the expectation is taken over the distribution of sequences induced by $\pi$ and $\rho$.

For simplicity, we omit the deviation plan and write 
$Q^{\pi} = Q^{\pi,\perp}$ and 
$V^{\pi} = V^{\pi,\perp}$,
for the special plan $\perp$ such that $\perp(\sigma,s,a) = a$ for all $\sigma,s,a$, which effectively means no deviation. 
With these notions, we define a CE as follows. 

\begin{definition}[$\delta$-CE]
\label{def:delta-IC}
A policy $\pi$ forms a {\em $\delta$-CE} (or {\em CE} for $\delta=0$), if and only if, for every agent $i \in \{1,\dots,n\}$ and every sequence $(\sigma; s, a) \in \Sigma \times S \times A$, the following condition holds for all possible deviation plans $\rho: \Sigma \times S \times A \to A$:
\begin{align}
\label{eq:IC-Q}
\sum_{\va \,:\, a_i = a}\, 
\pi( \va \given \sigma, s) \cdot
Q_i^{\pi}(\sigma;s,\va) 
\;\ge\;
\sum_{\va \,:\, a_i = a}\, 
\pi( \va \given \sigma, s) \cdot
Q_i^{\pi,\rho} (\sigma;s,\va) - \delta.
\end{align} 
\end{definition}

In other words, \Cref{eq:IC-Q} describes an IC condition based on the Q-values: whenever agent~$i$ is recommended to play $a$, they have no incentive (up to a $\delta$ tolerance) to deviate, given their belief about actions played by the other players conditioned on $a$.
The quantities on the two sides of \Cref{eq:IC-Q} are equal to agent~$i$'s expected values conditioned on being recommended $a$ (as \Cref{eq:V} conditioned on $a_i = a$), scaled by the marginal probability $\sum_{\va \,:\, a_i = a} \pi( \va \given \sigma, s)$. 
The value induced by a CE is defined as follows.

\begin{definition}[$\delta$-inducibility]
\label{def:delta-inducibility}
A value $\vv \in \mathbb{R}^{n+1}$ is said to be {\em induced by} a $\delta$-CE $\pi$ at state $s$ if $\vv = V^\pi(s)$.%
\footnote{For simplicity, we will often ombit $\varnothing$ and write $V^\pi(s) = V^\pi(\varnothing; s)$.} 
It is {\em $\delta$-inducible} (or {\em inducible} for $\delta=0$) if it can be induced by some $\delta$-CE.
\end{definition}

\paragraph{Optimal CE}
Our goal, from the principal's perspective, is to find a $\delta$-CE (ideally, $\delta=0$) that induces the maximum possible value for the principal (i.e., player~$0$) at the initial state $s_\init$; that is,
\[
\pi \in \argmax_{\pi' \in \calC_\delta} V_0^{\pi'}(s_\init).
\]
where $\calC_\delta$ denotes the set of all $\delta$-CEs.

Since $\pi$ defines a distribution for every sequence $(\sigma;s)\in \Sigma \times S$, it cannot be efficient if we aim to obtain an explicit representation of $\pi$, which lists every $\pi(\sigma;s)$. 
Hence, we consider the problem of computing $\pi(\sigma;s)$ for any given $(\sigma;s)$, and we measure the time complexity in terms of the size of the game as well as the length of the input sequence $\sigma$.
An algorithm that solves this problem allows one to execute $\pi$ while computing the distributions required for the execution on the fly. 

As we will demonstrate, optimal policies are inapproximable if we restrict probabilities in the policy to {\em rational} numbers. We discuss this issue in the next section and revise our objective accordingly to focus on rational representations. 

Similarly to the definition of a $\delta$-CE, we consider additive approximation errors in this paper, while the results can be extended to multiplicative approximation ratios too when all rewards are non-negative.
For additive errors to be meaningful, we assume that all rewards are normalized to be within the interval $[\rmin, \rmax] = [-\frac{1-\gamma}{2},\frac{1-\gamma}{2}]$ (otherwise, errors can be scaled up arbitrarily). 
This normalization is without loss of generality given that we consider a constant $\gamma$;
meanwhile, it ensures that the principal's cumulative reward is bounded in the constant interval $[-\frac{1}{2}, \frac{1}{2}]$ of length $1$.

\section{Inapproximability and Resource-Augmented Approximation}
\label{sc:inapprox-rational}

Unlike in finite-horizon games, when the horizon is infinite, irrational probabilities may arise in solutions to stochastic games (as noted by \citet{shapley1953stochastic}).
If there are other feasible solutions in the neighborhood of an irrational solution, we could potentially trade off a small $\epsilon$ in the principal's value for a near-optimal solution involving only rational numbers.
Unfortunately, as we demonstrate in \Cref{prp:irrational}, there are instances that do not admit any such near-optimal {\em rational} solutions. The gap between irrational and rational solutions can be arbitrarily large, making optimal $\delta$-CEs {\em completely inapproximable} by using rational numbers.

\begin{restatable}{theorem}{thmirrational}
\label{prp:irrational}
For any $\delta \in [0, 1)$, there exists a two-agent game in which $\max_{\pi \in \calC_\delta} V_0^{\pi}(s_\init) = 1/2$ and $\max_{\pi \in \calC_\delta^{\mathbb{Q}}} V_0^{\pi}(s_\init) = -1/2$, where $\calC_\delta^{\mathbb{Q}}$ denotes the set of $\delta$-CEs in which the probabilities are restricted to be in $\mathbb{Q}$.
\end{restatable}

To prove \Cref{prp:irrational}, we construct the game depicted in \Cref{fig:irrational}.
In this game, two agents each control one of the states, $s_1$ or $s_2$.
The rewards are structured such that the principal gets value $1/2$ only when both agents play $\ain$, whereby the game proceeds to the subgame rooted at $s_3$; and the principal gets $-1/2$, otherwise.
In order to attract both agents to play $\ain$, the principal needs to tune their values to an appropriate point where both agents' values exceed what can be attained for playing $\aout$ (which leads to the subgames rooted at $s'_1$ and $s'_2$).
The rewards in the subgames at $s'_1$ and $s'_2$ are further designed in a way such that the lowest inducible values for both agents are irrational numbers. 
Moreover, only when the agents' values in the subgame at $s_3$ match precisely these irrational numbers, can we simultaneously induce both agents to play $\ain$ since the agents' values in the subgame at $s_3$ are set to be negatively correlated.
Consequently, policies involving only rational probabilities cannot induce both agents to play $\ain$.

\tikzstyle{every state}=[very thick, fill=mywhite, minimum size=8mm] 

\begin{figure}[t]
\centering
\tikzstyle{subgamenode}=[draw=black!50,fill=black!5, text=black]
\tikzstyle{principal}=[draw=myblue,fill=myblue!5, text=myblue]

\begin{tikzpicture}[baseline={(current bounding box.center)}, node distance=22mm, scale=0.85, transform shape]
\tikzstyle{every node}=[text width=6mm, inner sep=1pt, align=center, font=\small] 

\node[state, initial, initial where=above] (s1) {$s_1$};
\node[state, subgamenode, below left of=s1, xshift=-3mm] (sout1) {$s'_1$};
\node[state, below right of=s1, xshift=3mm] (s2) {$s_2$};
\node[state, subgamenode, below left of=s2, xshift=-3mm] (sout2) {$s'_2$};
\node[state, below right of=s2, xshift=3mm, principal] (sin) {$s_3$};
\node[state, subgamenode, below of=sout1] (t1) {$s''_1$};
\node[state, subgamenode, below of=sout2] (t2) {$s''_2$};
\node[state, below left of=sin, xshift=6mm, principal] (t3) {$s'_3$};
\node[state, below right of=sin, xshift=-6mm, principal] (t4) {$s''_3$};
\draw 
(s1) edge[] node[labelOnEdge,text width=7mm]{$\aout$} (sout1)
(s1) edge[] node[labelOnEdge,text width=7mm]{$\ain$} (s2)
(s2) edge[] node[labelOnEdge,text width=7mm]{$\aout$} (sout2)
(s2) edge[] node[labelOnEdge,text width=7mm]{$\ain$} (sin)
(sout1) edge[subgamenode] node[labelOnEdge,text width=12mm]{$ab,\,ba$} (t1)
(sout1) edge[subgamenode,fill=none,in=180,out=-140,looseness=10] node[left, text width=11mm]{$aa,\,bb~$} (sout1)
(sout2) edge[subgamenode] node[labelOnEdge,text width=12mm]{$ab,\,ba$} (t2)
(sout2) edge[subgamenode,fill=none,in=0,out=-40,looseness=10] node[right, text width=11mm]{$~aa,\,bb$} (sout2)
(sin) edge[principal] node[labelOnEdge,text width=4mm]{$a$} (t3)
(sin) edge[principal] node[labelOnEdge,text width=4mm]{$b$} (t4)
(t1) edge[in=-105,out=-75,looseness=11,subgamenode,fill=none] (t1)
(t2) edge[in=-105,out=-75,looseness=11,subgamenode,fill=none] (t2)
(t3) edge[in=-105,out=-75,looseness=11,principal,fill=none] (t3)
(t4) edge[in=-105,out=-75,looseness=11,principal,fill=none] (t4)
;

\end{tikzpicture}
\vspace{3mm}
\caption{The game instance for proving \Cref{prp:irrational}.}
\label{fig:irrational}
\end{figure}
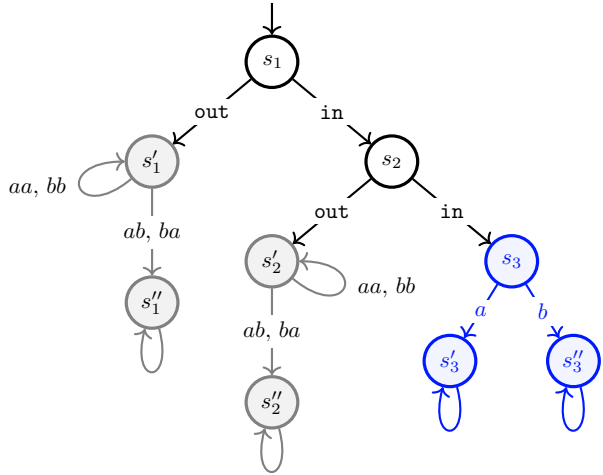

\paragraph{$(\epsilon,\delta)$-Optimality}

The gap $1$ between the two V-values in \Cref{prp:irrational} indicates that no algorithm based on rational representation of policies is guaranteed to return a $\delta$-CE whose value is smaller than that of a best possible $\delta$-CE by at most $\epsilon$, for any $\epsilon \in [0,1)$.
Given this inapproximability, we relax our benchmark to the value yielded by an optimal {\em exact} CE, when we define the $\epsilon$-optimality of a $\delta$-CE.
This approach is the same as {\em resource augmentation} in approximation algorithm design,
whereby we allow the algorithm to use $\delta$ more resources but still measure its performance against the best possible solution with the original, unaugmented resources.
More formally, we define our revised objective, $(\epsilon,\delta)$-optimal CEs, as follows.

\begin{definition}[$(\epsilon,\delta)$-optimal CE]
\label{def:eps-delta-optimal-policy}
A policy $\pi$ is an {\em $(\epsilon,\delta)$-optimal CE} if $\pi \in \calC_{\delta}$ and $V_0^\pi(s_{\init}) \ge \max_{\pi' \in \calC_0} V_0^{\pi'}(s_{\init}) - \epsilon$. 
\end{definition}

In the next section, we will introduce an algorithm that computes an $(\epsilon,\delta)$-optimal CE in polynomial time for a constant number of agents.
As mentioned, key to the tractability is the use of history-dependent policies. In contrast, optimal stationary CEs, not only offer lower values but are also intractable even with the above relaxed approximation criterion, according to the following result.

\begin{restatable}{theorem}{thminapproximability}
\label{thm:inapproximability}
There exist constants $\epsilon, \delta > 0$ such that,  
unless P = NP, no algorithm is guaranteed to compute an $(\epsilon,\delta)$-optimal {\bf\em stationary} CE in polynomial time.
The result holds for any discount factor $\gamma \in (0,1)$, even when there are only two agents, each player has at most four actions, and the horizon length is three.%
\hideinrestate{\footnote{It should be evident from the proof that this inapproximability result also holds for Markovian policies (see \cref{fn:markovian}), since in the reduced game instance each state is only reachable at a unique time step.}}
\end{restatable}

\section{Polynomial-Time Algorithm for Constant Number of Agents}

\label{sc:value-set-iteration}

We now present an algorithm for computing an $(\epsilon,\delta)$-optimal CE, which runs in polynomial time for a constant number of agents.
At a high level, the algorithm generalizes value iteration for solving MDPs, which iteratively updates estimates of state values. However, unlike MDPs---which allow for optimal {\em stationary} policies---our multi-player setting requires non-stationary policies for achieving optimality. 
It is therefore insufficient to associate only a single value with each state as the game may proceed differently from the same state when the history differs.
Instead, we evaluate, for each state, the entire set of values that can be induced when the game starts at this state. We call these sets the {\em inducible value sets}.

\begin{definition}[Inducible value set]
\label{def:inducible-value-set}
The {\em inducible value set} $\calV^\star(s) \subseteq \mathbb{R}^{n+1}$ of each state $s \in S$ consists of {\em all} values that are inducible at $s$.
\end{definition}

Key to our approach is a characterization of $\calV^\star$ as a fixed point of a map between functions of the form $\calV: S \rightrightarrows \mathbb{R}^{n+1}$. We refer to such functions as {\em value-set functions}.

\subsection{Fixed-Point Characterization of $\calV^\star$}

Let us first characterize each point $\vv \in \calV^\star(s)$ using the following constraints.

\begin{itemize}
\item 
{\bf Bellman constraint.}
By definition, $\vv \in \calV^\star(s)$ means that it can be induced by some CE $\pi$, such that $\vv = V^\pi(s)$.
Consider the first time step. We can expand $\vv$ according to \Cref{eq:Q,eq:V} as:
\begin{align}
\label{eq:const-bellman}
\vv = \bellman(s, \bar{\pi}, \ww) 
\coloneqq
\E_{\va \sim \bar{\pi}} \Big( \rr(s, \va) + \gamma \E_{s' \sim p(\cdot \given s, \va)} \ww(\va, \va; s') \Big), 
\end{align}
where $\bar{\pi} = \pi(s) \in \Delta(\bfA)$ captures the joint action distribution defined by $\pi$ in the first time step; and
each $\ww(\va, \va; s') = V^{\pi}(s, \va, \va; s') \in \mathbb{R}^{n+1}$ captures the onward value induced by $\pi$ at the next time step.
Hence, $(\bar{\pi}, \ww)$ can be viewed as a compact representation of $\pi$.
The two $\va$'s in $\ww(\va, \va; s')$---representing actions recommended and actions played, respectively---are the same because $\pi$, as a CE, is incentive compatible.
In terms of $\bar{\pi}$ and $\ww$, this means that the following IC constraint must be satisfied, following \Cref{eq:IC-Q}.

\item 
{\bf IC constraint.}
For every agent $i \in \{1,\dots,n\}$, every recommendation $a \in A$ to this agent, and every possible immediate derivation $b \in A$:
\begin{align}
& \label{eq:const-ic}
\sum_{\va \in \bfA:\, a_i = a} \bar{\pi}( \va) \cdot \Big( r_i\left(s, \va \right) + \gamma \E_{s' \sim p(\cdot \given s, \va )}\, w_i \left(\va, \va; s' \right) \Big) \ge \nonumber\hspace{-30mm}\\
&\qquad \sum_{\va \in \bfA:\, a_i = a} \bar{\pi}( \va) \cdot \Big( r_i\left(s, \va \oplus_i b \right) + \gamma \E_{s' \sim p(\cdot \given s,\, \va \oplus_i b )}\, w_i \left(\va, \va \oplus_i b; s' \right) \Big).
\end{align}
It suffices to consider such immediate deviations instead of deviation plans defined over the entire time horizon (as in \Cref{def:delta-IC}) because $\pi$ is IC in every subgame: the best a deviating agent can achieve in the subsequent time steps does not exceed what they obtain by following $\pi$, which is encoded in $\ww$.
Indeed, this also means that the onward values must fall inside the respective inducible value sets, so we have the following constraint.

\item 
{\bf Onward value constraint.}
For every $(\va, \bb; s') \in \bfA^2 \times S$:
\begin{align}
\label{eq:const-onward}
\ww(\va, \bb; s') \in \calV^\star(s').
\end{align}
\end{itemize}

Effectively, \Cref{eq:const-ic,eq:const-onward} together define the set of $(\bar{\pi}, \ww)$ tuples corresponding to CEs, while \Cref{eq:const-bellman} defines the value induced by $(\bar{\pi}, \ww)$.
More formally, the following proposition confirms the correctness of the above characterization.

\begin{restatable}{proposition}{prpVstar}
\label{prp:V-star}
$\vv \in \calV^\star(s)$ if and only if $\vv = \bellman(s, \bar{\pi}, \ww)$ for some $\bar{\pi} \in \Delta(\bfA)$ and $\ww: \bfA^2 \times S \to \mathbb{R}^{n+1}$ 
satisfying \Cref{eq:const-ic,eq:const-onward}.
\end{restatable}

Hence, $\calV^\star(s)$ consists of all values that can be characterized via \Cref{eq:const-bellman,eq:const-ic,eq:const-onward}, while \Cref{eq:const-onward} involves $\calV^\star$ itself. 
We can then view $\calV^\star$ as a fixed point of this characterization.
Formally, let us define the following map $\Phi$ from one value-set function to another: for each $s \in S$,
\begin{align}
\label{eq:Phi}
\Phi(\calV)(s) \coloneqq 
\Big\{ 
\bellman(s, \bar{\pi}, \ww)
\;\Big|\;
(\bar{\pi}, \ww) \in \calF_s(\calV)
\Big\},
\end{align} 
where we denote by $\calF_s(\calV)$ the set of tuples $(\bar{\pi}, \ww)$ satisfying \Cref{eq:const-ic,eq:const-onward}, with $\calV$ in place of $\calV^\star$ in \Cref{eq:const-onward}, i.e.,
\[
\calF_s(\calV) \coloneqq
\Big\{
(\bar{\pi}, \ww) 
\,\Big|\,
\begin{array}{l}
\bar{\pi} \in \Delta(\bfA),\, 
(\bar{\pi}, \ww) \text{ satisfies \Cref{eq:const-ic}, and } \ww(\va,\bb;s') \in \calV(s') \text{ for all } \va, \bb, s'
\end{array}
\Big\}.
\]

It then follows immediately from \Cref{prp:V-star} that $\calV^\star = \Phi(\calV^\star)$.
The following lemma further shows that $\calV^\star$ is not only a fixed point of $\Phi$ but also a {\em greatest} one.

\begin{restatable}{lemma}{lmmPhifixedpoint}
\label{prp:Phi-fixed-point}
If $\calV = \Phi(\calV)$, then $\calV \subseteq \calV^\star$.%
\hideinrestate{\footnote{For convenience, for any value-set functions $\calV$ and $\calV'$, and any set $X \subseteq \mathbb{R}^{n+1}$, we write $\calV' \subseteq \calV$ if $\calV'(s) \subseteq \calV(s)$ for all $s \in S$; and $\calV \subseteq X$ if $\calV(s) \subseteq X$ for all $s \in S$.}}
\end{restatable}

By definition, \Cref{prp:Phi-fixed-point} means that every $\vv \in \calV(s)$, where $\calV$ is an {\em arbitrary} fixed point of $\Phi$, is inducible at $s$.
Intuitively, this holds because $\vv$ can be expanded into a set of onward values given by $\ww$, while the fact that $\ww(\cdot,\cdot;s') \in \calV(s')$ ensures that these onward values can be expanded further in the same way.  
This expansion process can then continue indefinitely with $\calV$ being an invariant, and one can show that the expected cumulative reward yielded is equal to $\vv$.
We will later prove a more general version of this lemma (see \Cref{prp:eps-neigh-inducible}) necessary for our algorithm design.

Additionally, $\Phi$ preserves convexity, closure, and nonemptiness. These properties are essential for proving the convergence of value-set iteration to $\calV^\star$, which we discuss next.

\begin{restatable}{lemma}{lmmPhiproperties}
\label{lmm:Phi-properties}
$\Phi(\calV)(s)$ is convex, closed, and nonempty if $\calV(s')$ is convex, closed, and nonempty for all $s' \in S$. 
\end{restatable}

\subsection{Value-Set Iteration}

The map $\Phi$ leads to a {\em value-set iteration} process, producing a sequence $\calV^0, \calV^1, \dots$ of value-set functions such that 
\[
\calV^{k+1} = \Phi(\calV^k)
\]
for $k \in \mathbb{N}$.
Ideally, we would like the sequence to behave similarly to the classical value iteration method for solving MDPs, where convergence to a {\em unique} fixed point is guaranteed from any initial point, following the contraction mapping theorem (Banach fixed-point theorem).
However, the IC constraints in our multi-player setting complicate the situation, making it substantially different from standard value iteration.%
\footnote{Without these constraints, our problem becomes a multi-objective MDP \cite{chatterjee2006markov}.}

As it turns out, the sequence $\calV^0, \calV^1, \dots$ need not be convergent for any $\calV^0$. Indeed, $\Phi$ does not satisfy the {\em contraction} property required by the contraction mapping theorem.
Even when the sequence converges, it may converge to a fixed point different from $\calV^\star$ (in which case the point must be a strict subset of $\calV^\star$ according to \Cref{prp:Phi-fixed-point}), hence failing to identify all inducible values, especially those optimal for the principal.\footnote{For example, in the instance in \Cref{fig:intro}, starting with $\calV^0(s) = \{(0,0)\}$ for all $s \in S$ will result in $\calV^k(s) = \{(0,0)\}$ for all $k \in \mathbb{N}$, while $(3,3)$ is also inducible via the equilibrium in which both players play $C$ throughout.}
To proceed, we turn to the following monotonicity property of $\Phi$ (\Cref{lmm:Phi-monotonicity}).

\begin{restatable}[Monotonicity]{lemma}{lmmPhimonotonicity}
\label{lmm:Phi-monotonicity}
If $\calV' \subseteq \calV$, then $\Phi(\calV') \subseteq \Phi(\calV)$.
\end{restatable}

By Tarski's fixed point theorem \cite{tarski1955lattice}, a greatest fixed point exists for such monotonic maps.
In the case of $\Phi$, we further show that this greatest fixed point can be obtained in the limit of the value-set iteration sequence, as long as we initialize $\calV^0$ to the hypercube $\calB = \left[-1/2,\, 1/2 \right]^{n+1}$, which contains all possible inducible values.%
\footnote{Recall that all rewards are in the range $[\rmin, \rmax] = [-\frac{1-\gamma}{2},\frac{1-\gamma}{2}]$.}
Starting from there, $\Phi$ will produce a sequence of nested sets, each contained in their predecessor (\Cref{lmm:nested-sets}). 
Eventually, the sequence converges to $\calV^\star$ as we show in \Cref{thm:Phi-converge}, where we also prove the nonemptiness of $\calV^\star$.%
\footnote{Note that a sequence of nested sets, by itself, does not guarantee convergence to a nonempty set. E.g., the nonempty open interval $(0,1/k)$ does not converge to any nonempty set as $k\to \infty$.
Furthermore, even when it converges to a nonempty set, the set is not necessarily a fixed point of the map.
In the proof of \Cref{thm:Phi-converge}, we demonstrate that, thanks to the closure-preserving property of $\Phi$ (\Cref{lmm:Phi-properties}), convergence to a nonempty fixed point is guaranteed.
}

\begin{restatable}{lemma}{lmmnestedsets}
\label{lmm:nested-sets}
If $\calV^0 = \calB$, then 
$\calV^k \supseteq \calV^{k+1} \supseteq \calV^\star$ 
for all $k \in \mathbb{N}$.
\end{restatable}

\begin{restatable}{theorem}{thmPhiconverge}
\label{thm:Phi-converge}
If $\calV^0 = \calB$, then $\calV^k \to \calV^\star$ as $k \to \infty$, i.e.,  
for any $s \in S$, $\vv \in \calV^\star(s)$, and $\vv' \notin \calV^\star(s)$, there exists $k_0 \in \mathbb{N}$ such that $\vv \in \calV^k(s)$ and $\vv' \notin \calV^k(s)$ for all $k \ge k_0$.
Moreover, $\calV^\star = \bigcap_{k=1}^\infty \calV^k \neq \emptyset$.
\end{restatable}

Since convergence to $\calV^\star$ may require infinitely many iterations, a practical approach is to terminate the process when $\calV^k$ is sufficiently close to $\calV^\star$.
This is why a $\delta$ slack in the incentive constraints is necessary for our computational approach.
Moreover, for the approach to be tractable, we will maintain an approximation of $\calV^k$ because the space required for representing the exact $\calV^k$ can grow exponentially as $k$ increases.
We introduce an approximate version of $\Phi$ next.

\subsection{Approximately Fixed Point}
Hereafter, we let $\xi = (1-\gamma) \cdot \min \{\epsilon,\, \delta/2\}$
and assume w.l.o.g. that $1/2$ is a multiple of $\xi$ to sidestep trivial rounding issues.
We define the following approximate map $\xPhi$.
For all $s \in S$,
\[
\xPhi(\calV)(s) \coloneqq
\conv\Big(
\neigh_\xi \big(\Phi(\calV)(s) \big)
\;\cap\; G_\xi
\Big).
\]
Here, $\conv(\cdot)$ denotes the convex hull; 
$G_\xi = \{\xi \cdot \xx \mid \xx \in \mathbb{Z}^n \}$ is the set of grid points whose coordinates are multiples of $\xi$;
and
$\neigh_\xi(X) \coloneqq \bigcup_{\xx \in X} \neigh_\xi(\xx)$ is the $\xi$-neighborhood of a set $X$, where we abuse notation and let $\neigh_\xi(\xx) = \left\{ \xx' \in \mathbb{R}^{n+1} \,\middle|\, \|\xx' - \xx \|_\infty \le \xi \right\}$ for every $\xx \in \mathbb{R}^{n+1}$.

Namely, $\xPhi(\calV)(s)$ is the convex hull of those grid points in the $\xi$-neighborhood of  $\Phi(\calV)(s)$. 
The following remarks can help to understand this construction.%
\begin{itemize}
\item 
Using the grid points allows us to later reduce the problem of computing $\xPhi(\calV)(s)$ to checking if any given point is inside $\neigh_\xi \big(\Phi(\calV)(s) \big)$. 

\item Expanding $\Phi(\calV)(s)$ to its neighborhood ensures that $\xPhi(\calV)(s)$ contains every point in $\Phi(\calV)(s)$.
This prevents value-set iteration from collapsing to sets strictly smaller than those in $\calV^\star$. 
\end{itemize}

The following lemma shows that $\xPhi$ effectively approximates $\Phi$: it is bounded between $\Phi$ and its $\xi$-neighborhood.

\begin{restatable}{lemma}{lmmxPhiPhi}
\label{lmm:xPhi-Phi}
$\Phi(\calV) \subseteq \xPhi(\calV) \subseteq \neigh_\xi(\Phi(\calV))$ if $\calV(s)$ is convex for every $s \in S$.%
\hideinrestate{\footnote{For any $\calV: S \rightrightarrows \mathbb{R}^{n+1}$, by $\neigh_\xi(\calV)$ we mean the value-set function consisting of the $\xi$-neighborhood of each $\calV(s)$, i.e., $\neigh_\xi(\calV) =  \left( \neigh_\xi (\calV(s)) \right)_{s \in S}$.}}
\end{restatable}

More crucially, by noting that $\xPhi(\calV)(s)$ is convex by construction, we have the following corollary: any fixed point of the approximate map $\xPhi$ is {\em approximately fixed} under the original map $\Phi$.

\begin{corollary}
\label{crl:approx-fixed}
If $\calV = \xPhi(\calV)$, then $\Phi(\calV) \subseteq \calV \subseteq \neigh_\xi(\Phi(\calV))$.
\end{corollary}

\setlength{\algomargin}{1em} 
\begin{algorithm}[t]
\caption{Value-set iteration \label{alg:value-set-iteration}}
\setstretch{1.2}

$\calV(s) \leftarrow \neigh_\xi(\calB)$ for all $s \in S$%
\tcp*{$\calB = [-1/2, \, 1/2]^{n+1}$}

\lWhile{$\calV \neq \xPhi(\calV)$}
{$\calV \leftarrow \xPhi(\calV)$}
\Return $\calV$;

\end{algorithm}

We then aim to obtain a fixed point of $\xPhi$ and we do so by running value-set iteration, using $\xPhi$ as the update map as described in \Cref{alg:value-set-iteration}.
Indeed, $\xPhi$ is also monotonic and satisfies the properties in \Cref{lmm:Phi-properties}. 
Moreover, it always maps $\neigh_\xi(\calB)$ to its subset (given that $\xi$ divides $1/2$), so, similarly to the case with $\Phi$, initializing the value sets to $\neigh_\xi(\calB)$ results in a sequence of nested sets converging to a fixed point of $\xPhi$.

\Cref{prp:eps-neigh-inducible} further shows that every fixed point of $\xPhi$ can be induced approximately.
Here, we extend our previous inducibility notion to $(\epsilon,\delta)$-inducibility, which allows an $\epsilon$ offset in the induced value.  
The lemma generalizes \Cref{prp:Phi-fixed-point}, with the latter being the case where $\xi = 0$.

\begin{definition}[$(\epsilon,\delta)$-inducibility]
A value vector $\vv \in \mathbb{R}^{n+1}$ is {\em $(\epsilon,\delta)$-inducible} at state $s$ if some $\tilde{\vv} \in \neigh_\epsilon(\vv)$ is $\delta$-inducible at state $s$.
\end{definition}

\begin{restatable}{lemma}{thmepsneighinducible}
\label{prp:eps-neigh-inducible}
If $\calV = \xPhi(\calV)$, then every $\vv \in \calV(s_\init)$ is $(\epsilon,\delta)$-inducible at $s_\init$.%
\hideinrestate{\footnote{This also holds for every $s \in S$ by noting that the choice of $s_\init$ is arbitrary.}}
\end{restatable}

The proof of \Cref{prp:eps-neigh-inducible} follows the same intuition as that of \Cref{prp:Phi-fixed-point} but requires a more sophisticated construction. 
Intuitively, we can still induce each $\vv \in \calV(s)$ by expanding it into onward values, and iteratively expanding the onward values in the same way.
If $\calV$ were equal to $\Phi(\calV)$ as in \Cref{prp:Phi-fixed-point}, we could always find onward values within $\calV$ itself. 
But now that we only have $\calV = \xPhi(\calV)$, we cannot rule out the case where $\Phi(\calV) \subsetneq \calV$. Hence, we may need onward values outside of $\calV$ to expand $\vv$.
To prevent the onward values from drifting away arbitrarily from $\calV$, we use a more controlled procedure that forces the values back into $\Phi(\calV)$ in every expansion step.

\let\oldnl\nl
\newcommand{\nonl}{\renewcommand{\nl}{\let\nl\oldnl}}

\begin{algorithm}[t]
\caption{Computing $\pi(\sigma; s)$ of a policy $\pi$ that induces $\vv$ approximately
\label{alg:compute-pi}}
\setstretch{1.2}

\SetKwInOut{Input}{input}
\Input{%
$\calV$ such that $\calV = \xPhi(\calV)$%
\tcp*{so $\Phi(\calV) \subseteq \calV \subseteq \neigh_\xi(\Phi(\calV))$ by \Cref{crl:approx-fixed}}\\ 
\nonl 
$\sigma = (s^1,\va^1, \bb^1;\dots;s^\ell,\bb^\ell, \va^\ell)$ and $s \in S$; \\ 
\nonl 
a vector $\vv \in \calV(s^1)$ to be induced.
}

\smallskip

$\vv^1 \leftarrow \vv$, and $s^{\ell+1} \leftarrow s$;

\For{$t=1,\dots,\ell+1$}
{

Find $\tilde{\vv}^t \in \neigh_\xi(\vv^t)$ such that $\tilde{\vv}^t = \bellman(s^t, \bar{\pi}, \ww)$ for some $(\bar{\pi}, \ww) \in \calF_{s^t}(\calV)$;
~~~~~~~~~~~~~~~~
\tcp{so $\tilde{\vv}^t \in \Phi(\calV)(s^t)$ by \Cref{eq:Phi}}
\label{ln:compute-vv}

\smallskip

\lIf{$t \le \ell$}{$\vv^{t+1} \leftarrow \ww(\va^t, \bb^t; s^{t+1})$}
\label{ln:vv-t}

}

\Return $\bar{\pi}^{\ell+1}$;
\end{algorithm}

More formally, this procedure is described in \Cref{alg:compute-pi} (which will later be integrated into our computation method).
In the algorithm, we force each $\vv^t$ back into $\Phi(\calV)(s^t)$ by actually expanding a value $\tilde{\vv}^t$ in the neighborhood of $\vv^t$ (see \Cref{fig:computing-pi}).
So long as $\calV \subseteq \neigh_\xi(\Phi(\calV))$, we can always find such a $\tilde{\vv}^t$, that can be expanded into onward values inside $\calV$.
Iteratively, the procedure then continues with $\calV$ being an invariant, similarly to the procedure we used for proving \Cref{prp:Phi-fixed-point}.

The error introduced by forcing the values back into $\Phi(\calV)$ accumulates over time but is bounded by $\xi/(1-\gamma) < \min\{\epsilon, \delta/2\}$, thanks to discounting.
Overall, the actual value induced differs from the original target $\vv$ by at most $\epsilon$, while the IC constraints are violated by at most $\delta$.
Formalizing the argument, we can prove \Cref{lmm:eps-neigh-inducible-alg}, and in turn \Cref{prp:eps-neigh-inducible} as its corollary.

\begin{restatable}{lemma}{lmmepsneighinduciblealg}
\label{lmm:eps-neigh-inducible-alg}
Fix any $\calV = \xPhi(\calV)$ and $\vv \in \calV(s_\init)$ in the input to \Cref{alg:compute-pi} and let $\pi$ be the policy where $\pi(\sigma; s)$ is equal to the output of the algorithm, for every $(\sigma; s) \in \Sigma \times S$ starting at $s_\init$.
Then, $\pi$ is a $\delta$-CE and $V^\pi(s_\init) \in \neigh_\epsilon(\vv)$.
\end{restatable}

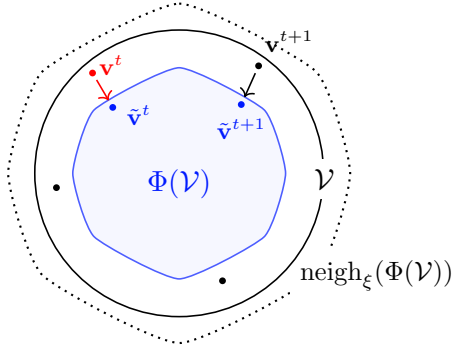
\begin{figure}
\centering
\begin{tikzpicture}
\tikzstyle{mydot} = [circle,fill,inner sep=0.9pt]
    \newcommand{\PA}{(-0.4, 0.3)}
    \newcommand{\PB}{( 0.0, 0.5)}
    \newcommand{\PC}{( 0.4, 0.3)}
    \newcommand{\PD}{( 0.5, 0.0)}
    \newcommand{\PE}{( 0.4,-0.3)}
    \newcommand{\PF}{( 0.0,-0.5)}
    \newcommand{\PG}{(-0.4,-0.3)}
    \newcommand{\PH}{(-0.5, 0.0)}
    \newcommand{\VA}{(-0.1, 0.2)}
    \newcommand{\VB}{( 0.0, 0.5)}
    \newcommand{\VC}{( 0.2, 0.7)}
    \newcommand{\VD}{( 0.5, 0.78)}
    \newcommand{\VE}{( 0.9, 0.6)}
    \newcommand{\VF}{( 1.05, 0.18)}
    \newcommand{\VG}{( 0.9,-0.2)}
    \newcommand{\VH}{( 0.5,-0.35)}
    \newcommand{\VI}{( 0.02,-0.18)}
    \newcommand{\vvt}{(-0.6, 0.7)}
    \newcommand{\tdvvt}{(-0.46, 0.46)}
    \newcommand{\vvtone}{(0.55, 0.75)}
    \newcommand{\tdvvtone}{(0.43,0.48)}
    \newcommand{\vvp}{(-0.85, -0.1)}
    \newcommand{\vvpp}{(0.3, -0.75)}

    \colorlet{vvtcolor}{red}
    \colorlet{onwardveccolor}{black}
    \colorlet{colorPhiVOutline}{myblue}
    \colorlet{colorPhiVFill}{myblue!5}

    \tikzstyle{onwardstyle} = [opacity=1,onwardveccolor]
    \tikzstyle{onwarddot} = [onwardstyle,mydot,semithick,fill=onwardveccolor,draw=none]

    \begin{scope}[scale around={4.5:(0.0,0.0)}]
        \draw [black, thick, dotted, fill=white] plot[smooth cycle, tension=0.3] coordinates
            {\PA \PB \PC \PD \PE \PF \PG \PH};
        \node[black,fill=white, xshift=8mm] at \PE {\small $\neigh_\xi(\Phi(\calV))$};
    \end{scope}

    \begin{scope}[scale around={1.9:(0.0,0.0)}]
        \draw[black,semithick](0,0) circle (1) node {};
        \node[black,xshift=2mm,yshift=3mm,fill=white] at \VG {$\calV$};
    \end{scope}

    \begin{scope}[scale around={2.8:(0.0,0.0)}]
        \draw [colorPhiVOutline!70, semithick, fill=colorPhiVFill!70] plot[smooth cycle, tension=0.3] coordinates
            {\PA \PB \PC \PD \PE \PF \PG \PH};
        \node[colorPhiVOutline] at (0, -0.05) {$\Phi(\calV)$};  
    \end{scope}  

    \begin{scope}[scale around={1.9:(0.0,0.0)}]
        \node[vvtcolor] at \vvt [mydot, label={[vvtcolor,yshift=1mm,xshift=-1mm]right:{\footnotesize$\vv^t$}}] {};
        \node[colorPhiVOutline] at \tdvvt [mydot, label={[colorPhiVOutline,yshift=-1mm,xshift=0mm]right:{\footnotesize$\tilde{\vv}^t$}}] {};
        \node[onwarddot] at \vvtone [label={[onwardstyle,xshift=4mm,yshift=0mm]above:{\footnotesize$\vv^{t+1}$}}] {};
        \node[onwarddot,colorPhiVOutline] at \tdvvtone [label={[colorPhiVOutline,xshift=0mm, yshift=0mm]below:{\footnotesize$\tilde{\vv}^{t+1}$}}] {};

        \node[circle,draw=none,inner sep=2pt] at \vvt (vvt) {};
        \node[circle,draw=none,inner sep=2pt] at \tdvvt (tdvvt) {};
        \node[circle,draw=none,inner sep=2pt] at \vvtone (vvtone) {};
        \node[circle,draw=none,inner sep=2pt] at \tdvvtone (tdvvtone) {};
        
        \draw [vvtcolor,semithick] (vvt) -- (tdvvt);
        \draw [onwardstyle,semithick] (vvtone) -- (tdvvtone);

        \node[onwarddot] at \vvp 
        {};
        \node[onwarddot] at \vvpp 
        {};
    \end{scope}
\end{tikzpicture}
\vspace{2mm}
\caption{\label{fig:computing-pi} 
In iteration $t$, a value $\tilde{\vv}^t$ in the $\xi$-neighborhood of the target value $\vv^t$ is expanded into a set of onward values $\ww$ (black dots).
In the next iteration, a new target value $\vv^{t+1}$ is chosen among the onward values according to the input sequence $(\sigma,s)$.}
\end{figure}

We can establish the following analogue of \Cref{thm:Phi-converge}, regarding the behavior of value-set iteration under $\xPhi$. 

\begin{restatable}{lemma}{prpalgiterationoutput}
\label{prp:alg-iteration-output}
\Cref{alg:value-set-iteration} terminates in at most $|S| \cdot (1/\xi+2)^{n+1}$ iterations and outputs a $\calV$ such that $\calV = \xPhi(\calV)$ and $\calV \supseteq \calV^\star$.
\end{restatable}

The inclusion $\calV \supseteq \calV^\star$ in the lemma can be verified by noting that $\calV^k \subseteq \xcalV^k$ holds for the sequences $(\calV^k)_{k=1}^\infty$ and $(\xcalV^k)_{k=1}^\infty$ generated under $\Phi$ and $\xPhi$, respectively.
Since each iteration removes at least one grid point from some $\calV(s)$, the number of iterations it takes to reach the termination condition is bounded by the number of grid points in the initial value set $\neigh_\xi(\calB)$, considering all the states.

\subsection{Computing an $(\epsilon,\delta)$-Optimal CE}
\label{sc:computing-put-together}

Hence, by running \Cref{alg:value-set-iteration}, we obtain a value-set function $\calV$ with the properties stated in \Cref{prp:alg-iteration-output}.
Since $\calV \supseteq \calV^\star$, the set $\calV(s_\init)$ contains an element whose value for the principal is as high as the best value in $\calV^\star(s_\init)$.
It then holds for any $\tilde{\vv}$ in the $\xi$-neighborhood of any $\vv \in \argmax_{\vv' \in \calV(s_\init)} v'_0$ that
\[
\tilde{v}_0 \ge \max_{\vv' \in \calV^\star(s_\init)} v'_0 - \epsilon.
\] 
Consequently, any $\delta$-CE inducing such a $\tilde{\vv}$ is $(\epsilon,\delta)$-optimal.
Indeed, now that $\calV = \xPhi(\calV)$, by \Cref{crl:approx-fixed}, $\calV$ satisfies the condition required by \Cref{alg:compute-pi}.
So, on input $\calV$ 
\Cref{alg:compute-pi} computes precisely a $\delta$-CE inducing $\tilde{\vv}$, as \Cref{lmm:eps-neigh-inducible-alg} states.

This leads to our approach to computing an $(\epsilon,\delta)$-optimal CE, summarized as follows:
\begin{itemize}
\item[1.] Run \Cref{alg:value-set-iteration} to obtain a fixed point $\calV = \xPhi(\calV)$.

\item[2.] 
Pick arbitrary $\vv \in \arg\max_{\vv' \in \calV(s_\init)} v'_0$.

\item[3.] Run \Cref{alg:compute-pi} with $\calV$ and $\vv$ to compute a policy $\pi$, which forms an $(\epsilon,\delta)$-optimal CE.
\end{itemize}

We analyze the time complexity of these procedures to conclude this section.

\paragraph{Time Complexity}
\Cref{prp:alg-iteration-output} already gives an upper bound on the number of iterations required by \Cref{alg:value-set-iteration}, and by design, \Cref{alg:compute-pi} terminates in $|\sigma|+1$ iterations.
Hence, for both algorithms, it remains to bound the time complexity of each iteration.
For \Cref{alg:value-set-iteration}, this amounts to the time it takes to compute $\xPhi(\calV)$; 
and for \Cref{alg:compute-pi}, it is about computing $(\bar{\pi},\ww)$ to implement some $\tilde{\vv}^t$ close to $\vv^t$ (\Cref{ln:compute-vv}).
We show that both problems reduce to the following decision problem and can be solved in polynomial time via a linear programming (LP) formulation. 
Roughly speaking, the LP solves for $\bar{\pi}$ and $\ww \in \calV$ satisfying \Cref{eq:const-bellman,eq:const-ic}.

\begin{restatable}{lemma}{lmmcomputeinducibility}
\label{lmm:compute-inducibility}
Suppose that we are given a state $s \in S$, a value $\vv \in \mathbb{R}^{n+1}$, and a value-set function $\calV$ where each $\calV(s')$, $s'\in S$, is a convex polytope in vertex representation involving at most $L$ vertices.
It can be decided in time $\poly\left(|S|, |\bfA|, L \right)$ whether $\vv \in \neigh_\xi(\Phi(\calV)(s))$.
Moreover, in the case where $\vv \in \neigh_\xi(\Phi(\calV)(s))$, one can compute a tuple 
$(\bar{\pi}, \ww) \in \calF_s(\calV)$ such that $\bellman(s, \bar{\pi}, \ww) \in \neigh_\xi(\vv)$
in time $\poly\left(|S|, |\bfA|, L \right)$.
\end{restatable}

Specifically, in \Cref{alg:value-set-iteration}, computing $\xPhi(\calV)(s)$ reduces to deciding, for each grid point $\vg \in G_\xi \cap \calB$, whether $\vg \in \neigh_\xi(\Phi(\calV)(s))$. 
The set of grid points inside $\neigh_\xi(\Phi(\calV)(s))$ gives a vertex representation of $\xPhi(\calV)(s)$.\footnote{Some of the points may not be vertices of the polytope, but so long as all the vertices of $\xPhi(\calV)(s)$ are included, the convex hull of the points gives $\xPhi(\calV)(s)$.}
As for \Cref{alg:compute-pi}, note that 
$\vv^t \in \neigh_{\xi}(\Phi(\calV)(s^t))$
because: $\vv^t \in \calV(s^t)$ by \Cref{ln:vv-t}, while $\calV \subseteq \neigh_\xi(\Phi(\calV))$ according to the input requirement.
Hence, we can invoke the second part of \Cref{lmm:compute-inducibility} to obtain $(\bar{\pi},\ww)$ that implements a desired $\tilde{\vv}^t$.

\begin{restatable}{theorem}{thmmainalg}
\label{thm:main-alg}
We can compute a fixed point $\calV$ of $\xPhi$, $\calV \supseteq \calV^\star$, in time $\poly\left(|S|, |\bfA|, (1/\xi)^{n+1} \right)$.
Given $\calV$, there exists an $(\epsilon,\delta)$-optimal CE $\pi$, such that for any given $(\sigma; s) \in \Sigma \times S$ we can compute $\pi(\sigma; s)$ in time $\poly\left(|S|, |\bfA|, (1/\xi)^{n+1},\, |\sigma| \right)$.
\end{restatable}

\section{Beyond Constant Number of Agents: $c$-Turn-Based Games}
\label{sc:beyond-constant}

The time complexity of our algorithm is exponential in $n$. What if $n$ is not a constant? Can we still find a polynomial time algorithm?
This question is meaningful only when the reward and transition functions are given in succinct representation; otherwise, when they are given in the matrix form---which explicitly enumerates the parameters for each $\va \in \bfA$---the input size is by itself already exponential in $n$.
However, it is known that under succinct representation, even in one-shot games, {\em optimal} CEs can be intractable \cite{papadimitriou2008computing}. 
Given this barrier, we restrict our attention to games where optimal CEs are tractable at least in the one-shot setting.
A typical case is turn-based games, where only one player acts at each state. 
We consider a more general variant of turn-based games, called {\em $c$-turn-based games}, which allows up to a constant number $c$ of players to act at each state. 

\begin{definition}[$c$-turn-based games]
A game is {\em $c$-turn-based} if there exist a constant $c$ and a set $I_s \subseteq \{0, 1,\dots,n\}$, $|I_s| \le c$, for every $s \in S$, such that: $\rr(s, \va) = \rr(s, \va')$ and $p(\cdot \given s, \va) = p(\cdot \given  s, \va')$ if $a_i = a'_i$ for all $i \in I_s$.
W.l.o.g., we assume that  $0 \in I_s$ for all $s \in S$ (so the principal always acts).
\end{definition}
 
Clearly, the size of the representation of a $c$-turn-based game grows only polynomially with $n$.
We next present an algorithm that runs in time polynomial in $n$ for such games.

\subsection{$\lambda$-Memory Meta-Game}

To describe the algorithm, we first convert the original game into a {\em $\lambda$-memory meta-game}, or {\em meta-game} for simplicity.
Let
\begin{align}
\label{eq:lambda}
\lambda = \log (\xi/4) / \log \gamma,
\end{align}
so that $(1+\xi) \cdot \gamma^\lambda < \xi/2$.
In the meta-game, each state---call it a {\em meta-state} to avoid confusion---encodes the current as well as the previous $\lambda$ states in the original game.

Formally, let $X = (S \cup \{*\})^{\lambda}$ be the meta-state space, where $*$ is a placeholder for the first $\lambda$ time steps, when there are fewer than $\lambda$ previous states.
Each meta-state $\xx \in X$ is a tuple $\xx = (x^{-\lambda}, x^{-\lambda+1}, \dots, x^{-1}, x^0)$, where $x^0 \in S$ is the current state as in the original game, and each $x^{-\ell} \in S$ represents the state $\ell$ time steps before. 

The meta-game starts at the initial meta-state $\xx_\init = (*,\dots,*,s_\init)$ and it shares the same action space $\bfA$ with the original game.
The rewards and transition probabilities in the meta-game, denoted $\tilde{\rr}$ and $\tilde{p}$, respectively, are defined as follows so that the meta-game is effectively equivalent to the original game, except that each meta-state also records $\lambda$ previous states:
\begin{itemize}
\item $\tilde{\rr}(\xx, \va) = \rr(x^0, \va)$ for all $\xx \in X$ and $\va \in \bfA$.

\item $\tilde{p}(\xx' \given \xx, \va) = p( {x'}^0 \given x^0, \va)$ for all $\va \in \bfA$, if ${x'}^{-\ell-1} = x^{-\ell}$ for all $\ell = 0,\dots, \lambda-1$; and $\tilde{p}(\xx' \given \xx, \va) = 0$, otherwise.
\end{itemize}
Since $\lambda$ is a constant, the size of the meta-game is polynomial in the size of the original game.

\subsection{Polynomial-Time Algorithm for $c$-Turn-Based Games}

Key to our approach is the following ``turn-based'' map $\tPhi$, which modifies $\xPhi$ using an extended neighborhood notion. 

\paragraph{Turn-Based Map}

For each $\xx \in X$, we define
\begin{align}
\label{eq:Phi-tb}
\tPhi(\calV)(\xx) \coloneqq
\conv\Big(
\neigh_{\xi/2}^{I_\xx} \big(\Phi(\calV)(\xx) \big)
\;\cap\; G_{\xi/2}
\Big),
\end{align}
where
\[
I_\xx \coloneqq \bigcup_{\ell=0}^\lambda I_{x^{-\ell}}
\]
consists of all the acting players in the $\lambda+1$ states in $\xx$; and $\neigh_{\xi/2}^{I_\xx}$ relaxes the neighborhood notion as follow:
\[
\neigh_{\xi/2}^{I_\xx} (\vv) \coloneqq 
\Big\{
\vv' \in \neigh_{\xi/2}(\calB) 
\;\Big|\;
|v'_i - v_i| \le \xi/2 \text{ for all } i \in I_\xx
\Big\}.
\]
Namely, it only restricts dimensions inside $I_\xx$ to the neighborhood of $\vv$ while relaxing those outside of $I_\xx$ to $\neigh_{\xi/2}(\calB)$ (which will be used as the initial point for value-set iteration).
Intuitively, we can assign arbitrary values to agents who have not been active for more than $\lambda$ time steps; the influence of the current step to them is small due to discounting. 
The additional errors introduced by this relaxation can be handled by using the smaller error bound $\xi/2$ in \Cref{{eq:Phi-tb}} (compared to $\xi$ in the definition of $\xPhi$).

It can be verified that $\tPhi$ behaves similarly to $\xPhi$ as in \Cref{lmm:Phi-monotonicity,lmm:xPhi-Phi}: it is monotonic and it is bounded between $\Phi$ and its $\xi/2$-neighborhood (under the new neighborhood definition).

\paragraph{Value-Set Iteration with $\tPhi$}

We now run value-set iteration (\Cref{alg:value-set-iteration}) with $\tPhi$ and extends our results in \Cref{sc:value-set-iteration}.

First, using the monotonicity of $\tPhi$, we can show that \Cref{alg:value-set-iteration} generates a sequence of nested sets converging to a fixed point of $\tPhi$ in a finite number of iterations, similarly to \Cref{prp:alg-iteration-output}.
Crucially, since each value set involves only a constant number of active dimensions (all the other dimensions require no updates), it can be represented by grid points in the space spanning only these effective dimensions. The number of grid points to be considered for each meta-state becomes a constant. Consequently, \Cref{alg:value-set-iteration} terminates in a polynomial number of iterations.

\begin{restatable}{lemma}{lmmtPhitermination}
\label{lmm:tPhi-termination}
Using $\tPhi$ and initializing $\calV(s)$ to $\neigh_{\xi/2}(\calB)$ for every $s \in S$, 
\Cref{alg:value-set-iteration} terminates in at most $|S|^{\lambda+1} \cdot (2/\xi+2)^{(\lambda+1) \cdot c}$ iterations and outputs $\calV$ such that $\calV = \tPhi(\calV)$ and  $\calV \supseteq \calV^\star$.
\end{restatable}

We can further replicate \Cref{prp:eps-neigh-inducible} to establish the $(\epsilon,\delta)$-induciblity of every value in the fixed point of $\tPhi$.

\begin{restatable}{lemma}{prpepsneighinducibleturnbased}
\label{prp:eps-neigh-inducible-turn-based}
If $\calV = \tPhi(\calV)$, then every $\vv \in \calV(\xx_\init)$ is $(\epsilon,\delta)$-inducible at state $s_\init$.
\end{restatable}

Specifically, to prove \Cref{prp:eps-neigh-inducible-turn-based}, we consider the policy $\pi$ computed by \Cref{alg:compute-pi}, where we use the new neighborhood notion $\neigh_{\xi/2}^{I_\xx}$ in place of $\neigh_\xi$. 
Since $\neigh_{\xi/2}^{I_\xx}$ relaxes dimensions outside of $I_\xx$, the value $\tilde{\vv}^t$ found in each iteration may differ from $\vv^t$ by up to $1+\xi$ in these dimensions (where $1+\xi$ is the length of $\neigh_{\xi/2}(\calB)$). 
Our key argument is that, agents corresponding to the relaxed dimensions must have not been active for more than $\lambda$ time steps.
Hence, the value difference introduced by the relaxation to the time steps in which these agents are active is bounded by $(1+\xi) \cdot \gamma^\lambda < \xi/2$.

The above results eventually lead to the following theorem for $c$-turn-based games.
When $c$ is a constant, this gives an algorithm that computes an $(\epsilon,\delta)$-optimal CE in time polynomial in $n$.

\begin{restatable}{theorem}{thmmainalgcturn}
\label{thm:main-alg-c-turn}
When the game is $c$-turn-based, we can compute a fixed point $\calV$ of $\tPhi$ in time $\poly\left(|S|^{\lambda+1},\, |A|^{c+1},\, (1/\xi)^{(\lambda+1) \cdot c},\, n \right)$.
Given $\calV$, there exists an $(\epsilon,\delta)$-optimal CE $\pi$,  such that for any given $(\sigma,s) \in \Sigma \times S$ we can compute $\pi(\sigma,s)$ in time $\poly\left(|S|^{\lambda+1},\, |A|^{c+1},\, (1/\xi)^{(\lambda+1) \cdot c},\, |\sigma|, \, n \right)$.
\end{restatable}

\section{Conclusion and Discussions}
\label{sec:conclusion}

We presented algorithms for computing optimal CEs in infinite-horizon multi-player stochastic games. 
Our algorithms achieve $(\epsilon,\delta)$-optimality, for any $\epsilon$ and $\delta$. 
We also established matching inapproximability results that indicate the optimality of our approach.
For the general model, the algorithms run in time polynomial in $(\epsilon \delta (1-\gamma))^{-(n+1)}$.
In the special case of $c$-turn-based games, we reduce the time complexity to a polynomial in $n$.

While our primary focus is on computing optimal CEs, another widely considered and easier task is finding just one CE. 
Our algorithms readily solve this problem, in time exponential in the number $n$ of agents. 
An interesting open question is therefore whether this exponential dependence can be avoided in succinctly represented games. 
In the one-shot setting, this is known to be possible for many succinctly represented games \cite{papadimitriou2008computing}.

Regarding {\em space} complexity, we can easily adapt the value-set iteration algorithm to remove the exponential dependence on $n$: instead of tracking the entire value sets, we maintain only a single point in each value set. Due to the monotonicity property of the operator, the iteration process will generate a sequence of points $\vv^0, \vv^1, \dots$, with each $\vv^k$ contained in the corresponding value set $\calV^k$ generated by value-set iteration. Consequently, $\vv^k$ will eventually enter $\mathcal{V}^\star$ as $k$ increases, hence forming a CE. 
The space complexity of this method is polynomial in $n$ since only one point is maintained throughout.

However, it remains open whether the {\em time} complexity of this adapted value iteration approach can also be bounded from above by a polynomial in $n$. 
Our current analysis only yields an exponential upper bound. It is also possible that this problem is complete in some complexity class whose relation with P remains unknown (e.g., PPAD). Notably, computing a {\em stationary} CE has been shown to be PPAD-hard \cite{daskalakis2023complexity}. In our case, the CE to be computed need not be stationary, so this known PPAD-hardness does not imply the same for computing a CE.

\bibliographystyle{ACM-Reference-Format}

\clearpage

\appendix

\setboolean{isrestating}{true}

\section{Omitted Proofs}

\subsection{Proofs in \Cref{sc:inapprox-rational}}

\thmirrational*

\begin{proof}

We describe an instance for $\delta = 0$ where the agents' rewards are not all in the normalized range $[-\frac{1-\gamma}{2}, \frac{1-\gamma}{2}]$. 
At the end of this proof, we will describe how this instance can be modified so that it applies to any $\delta \in [0,1)$ while all rewards are in the normalized range.

Let there be two agents and let the players' action set be $A = \{a,b, \ain, \aout\}$.
The states and the transition model are shown in \Cref{fig:irrational}. 
All transitions are deterministic. Specifically:
\begin{itemize}
\item Each state $s_i$, $i\in \{1,2\}$, is controlled solely by agent~$i$---the transition and reward functions are invariant with respect to the other players' actions.

\item Each state $s'_i$, $i\in \{1,2\}$, is controlled jointly by the principal and agent~$i$ (but not the other agent). 
Hence, from $s'_i$ the game enters a subgame between the principal and agent~$i$.
In \Cref{fig:irrational}, we use $aa$, $bb$, $ab$, and $ba$ to denote the joint actions of the principal and agent $i$ in these subgames. 

\item All the other states are controlled solely by the principal.
\end{itemize}

We define the rewards in a way such that the principal strictly prefers both agents to select $\ain$.
To achieve this, the policy needs to reduce agent $i$'s value at $s'_i$ as much as possible.
The subgame is then structured so that {\em the minimum inducible value of agent $i$ is an irrational number.} 
Eventually, to attract both agents to select $\ain$, the principal needs to set the agents' values at $s_3$ to irrational numbers, in order to match the minimum inducible values.

More specifically, the rewards are given as follows. (Recall that $\rmax = \frac{1-\gamma}{2}$ and $\rmin = - \frac{1- \gamma}{2}$.)
\begin{itemize}
\item 
For the principal (player $0$), 
we let all rewards be $\rmin$ except the following:
\begin{align*}
&
r_0(s_1, \ain) \;=\; r_0(s_2, \ain) \;=\; \rmax; \\
&
r_0(s_3, \cdot) \;=\; r_0(s'_3, \cdot) \;=\; r_0(s''_3, \cdot) \;=\; \rmax.
\end{align*}
This ensures that the principal gets: a high value $\frac{1}{1-\gamma} \cdot \rmax = 1/2$ if both agents select $\ain$ with probability $1$; and a low value $\frac{1}{1-\gamma} \cdot \rmin = - 1/2$, otherwise.

\item
For each agent $i \in \{1,2\}$, let all rewards be zero except the following: 
\begin{align*}
& 
r_1(s_3, a) = 1/\gamma^2, \qquad\;
r_1(s'_1, aa) = 1, \qquad \;\;\; 
r_1(s'_1, bb) = 2; \\
&
r_2(s_3, a) = - 1/\gamma, \qquad
r_2(s'_2, aa) = -1, \qquad
r_2(s'_2, bb) = -2.
\end{align*}
\end{itemize}

The rewards at $s'_i$ result in the {\em irrational} minimum inducible values stated in \Cref{lmm:min-spi}.
Additionally, the rewards at $s_3$ make the values of agents~$1$ and $2$ at $s_3$ negatively correlated: 
\[
V_1^\pi(s_3) = \pi(a \given s_3) / \gamma^2
\quad\text{and}\quad
V_1^\pi(s_3) = - \pi(a \given s_3) / \gamma.
\]
As a result, $\pi(a \given s_3) = \sqrt{11/3} - 1 \approx 0.915$ is the only situation in which the Q-value of $\ain$ matches that of $\aout$ for both agents.
Any value lower (respectively, larger) than this would disincentivize agent~$1$ (respectively, agent~$2$) from choosing $\ain$, in which case the principal only obtains the low value $-1/2$.

\begin{restatable}{lemma}{lmmminspi}
\label{lmm:min-spi}
When $\gamma = 1/2$, the minimum inducible value of agent~$1$ at $s'_1$ is $\sqrt{11/3} - 1$. 
The minimum inducible value of agent~$2$ at $s'_2$ is $-\sqrt{11/3} + 1$.
\end{restatable}

\begin{proof}
To analyze the minimum inducible values, we view the subgame starting at $s_i$ as a zero-sum game between the principal and agent~$i$, given by the following symmetric payoff matrices.
\begin{table}[h]
\renewcommand{\arraystretch}{1.5}
\centering
\begin{tabular}{ccc}
                         & $a$                      & $b$                      \\ \clineB{2-3}{2.5} 
\multicolumn{1}{l}{$a$} & \multicolumn{1}{V{2.5}c|}{$~~1~$} & \multicolumn{1}{cV{2.5}}{$~~0~$} \\ \cline{2-3} 
\multicolumn{1}{l}{$b$} & \multicolumn{1}{V{2.5}c|}{$~~0~$} & \multicolumn{1}{cV{2.5}}{$~~2~$} \\ \clineB{2-3}{2.5} 
                         & \multicolumn{2}{c}{$r_1(s'_1, \cdot)$}               
\end{tabular}
\hspace{2cm}
\begin{tabular}{ccc}
                         & $a$                      & $b$                      \\ \clineB{2-3}{2.5} 
\multicolumn{1}{l}{$a$} & \multicolumn{1}{V{2.5}c|}{$-1~$} & \multicolumn{1}{cV{2.5}}{$~0~$} \\ \cline{2-3} 
\multicolumn{1}{l}{$b$} & \multicolumn{1}{V{2.5}c|}{$0$} & \multicolumn{1}{cV{2.5}}{$-2$} \\ \clineB{2-3}{2.5} 
                         & \multicolumn{2}{c}{$r_2(s'_2, \cdot)$}               
\end{tabular}
    \label{tab:my_label}
\end{table}

The agents' minimum inducible values are equal to their {\em maximin values} in these games.
In particular, it suffices to analyze agent~$1$'s maximin value. 
The value of agent~$2$ follows readily by noting that in the zero-sum game given by $r_2(s'_2, \cdot)$, agent~$2$' essentially plays the role of agent~$1$'s opponent as in the zero-sum game given by $r_1(s'_1, \cdot)$.  
Hence, agent~$2$'s maximin value is the opposite of that of agent~$1$.

In zero-sum games, correlation signals do not help to improve the value of any player, so we can assume that the principal only selects her own actions, by using a policy $\pi(\cdot \given s'_1) \in \Delta(\{a,b\})$.
Let $x = \pi(a \given s'_i)$ and suppose that the agent responds by playing action $a$ with probability $y$.
In this case, the agent's value is
\begin{align*}
v(x, y) = 
& \; xy \cdot (1 + \gamma \cdot v(x, y) ) \;\, + \;  
x(1- y) \cdot (0 + \gamma \cdot 0 ) \;  + \\
& \;  (1-x)y \cdot (0 + \gamma \cdot 0 ) \;  +\; 
(1 - x) (1- y) \cdot (2 + \gamma \cdot v(x, y)),
\end{align*}
where we used the fact that the value at state $s''_1$ is $0$.
When $\gamma = 1/2$, This gives
\[
v(x, y) = \frac{xy + 2(1-x)(1-y)}{1 - xy/2 - (1-x)(1-y)/2}.
\]
It can be verified that $v(x,y)$ is monotonic w.r.t. $y$ for any $x \in [0,1]$. Hence, the value of the agent's best response, as a function of $x$, is $f(x) = \max \{ v(x, 0), v(x,1)\}$.
The minimum value of $f$ is obtained at the point where $v(x, 0) = v(x,1)$. Solving this equation for $x\in[0,1]$ gives $x=(7 - \sqrt{33})/2$ and $v(x, 0) = \sqrt{11/3} - 1$.
\end{proof}

It remains to show how the game instance can be modified so that it applies to any $\delta \in [0,1)$ while all rewards fall in the normalized range $[\rmin, \rmax]$.
We proceed as follows.

Denote by $r_i$ be the reward function of agent $i$ in the original instance and by $r'_i$ the modified reward function. 
Moreover, let $\alpha_1 > 0$ and $\alpha_2 < 0$ be two numbers sufficiently close to zero, such that:
\[
|\alpha_i \cdot r| \le  ( 1 - \delta ) \cdot (\rmax - \rmin)
\]
for all reward values $r$ in the original instance.

For agent~$1$, since the original rewards are all nonnegative, to ensure that the new rewards fall within $[\rmin,\rmax]$, we let:
\begin{itemize}
\item $r'_1 = \alpha_1 \cdot r_1 + \rmin + \delta \cdot (\rmax - \rmin)$ for all rewards in the subgame starting at $s'_1$, as well as the reward for $(s_1, \aout)$, which triggers this subgame; and

\item $r'_1 = \alpha_1 \cdot  r_1 + \rmin$ for all other rewards.
\end{itemize}
Similarly, for agent~$2$, the original rewards are all nonpositive, so we let:
\begin{itemize}
\item $r'_2 = \alpha_2 \cdot  r_2 + \rmax$ for all rewards in the subgame starting at $s'_2$, as well as the reward for $(s_2, \aout)$, which triggers this subgame; and

\item $r'_2 = \alpha_2 \cdot r_2 + \rmax - \delta \cdot (\rmax - \rmin)$ for all other rewards.
\end{itemize}

Hence, the new rewards are all within $[\rmin, \rmax]$.
Additionally, the terms $\delta \cdot (\rmax - \rmin)$ and $-\delta \cdot (\rmax - \rmin)$ create a cumulative difference of $\sum_{t=0}^{\infty} \gamma^t \cdot \delta \cdot (\rmax - \rmin) = \delta$ between the values at $s_i$ and $s_3$ for each agent $i$.
This gap cancels out the $\delta$ slack in $\delta$-CE.
Based on the minimum inducible values in \Cref{lmm:min-spi} scaled by $c_i$, we can verify that the principal can only get the high value $1/2$ by using the same policy in our analysis of the original instance, which involves irrational probabilities.
\end{proof}

\thminapproximability*

\newcommand{\betagapsat}{{\sc GAP 3-SAT}}
\newcommand{\termin}{\perp}

\begin{proof}
The proof is via a reduction from the {\betagapsat} problem. 
A {\betagapsat} instance is given by a 3CNF formula $\varphi$.
It is a yes-instance if $\varphi$ is satisfiable, and a no-instance if no assignment of the boolean variables in $\varphi$ satisfies more than a $\beta$ fraction of the clauses in $\varphi$.
It is know via the PCP theorem that there exists a constant $\beta \in (0,1)$ such that {\betagapsat} is NP-hard \cite{haastad2001some}.

\begin{figure}[t]
\centering
\tikzstyle{cstate}=[state, draw=myblue,fill=myblue!5, text=myblue]
\tikzstyle{xstate}=[state, draw=black,fill=mywhite]
\tikzstyle{minoredge}=[-,thick, dotted]

\begin{tikzpicture}[baseline={(current bounding box.center)}, node distance=25mm, scale=0.85, transform shape]
\tikzstyle{every node}=[text width=6mm, inner sep=1pt, align=center, font=\small] 

\node[state, initial, initial where=above] (sinit) {$s_\init$};
\node[cstate, below of=sinit, xshift=-4mm, yshift=5mm] (C2) {$C_2$};
\node[cstate, left of=C2, xshift=-4mm] (C1) {$C_1$};
\node[right of=C2, xshift=-3mm] (dots1) {$\dots$};
\node[cstate, right of=dots1, xshift=-3mm] (Ck) {$C_k$};
\node[myTermState, below left of=C1, xshift=-5mm, yshift=0mm] (t0) {};
\node[xstate, below of=C1, xshift=0mm] (x1) {$x_1$};
\node[xstate, right of=x1] (x2) {$x_2$};
\node[xstate, right of=x2] (x3) {$x_3$};
\node[xstate, right of=x3] (x4) {$x_4$};
\node[myTermState, below left of=x1, xshift=8mm] (t1) {};
\node[myTermState, below right of=x1, xshift=-8mm] (t2) {};
\draw 
(sinit) edge[dashed] (C1)
(sinit) edge[dashed] (C2)
(sinit) edge[dashed] (Ck)
(C1) edge[] node[labelOnEdge,text width=15mm,pos=0.4, xshift=-8mm, yshift=1mm]{$(\aout,\aout)$} (t0)
(C1) edge[] node[labelOnEdge,text width=12mm,pos=0.55]{$(a_1,*)$} (x1)
(C1) edge[] node[labelOnEdge,text width=12mm,pos=0.55]{$(a_2,*)$} (x2)
(C1) edge[] node[labelOnEdge,text width=12mm,pos=0.55]{$(\aout,a_3)$} (x4)
(C2) edge[minoredge] ($(C2)!0.2!(x1)$)
(C2) edge[minoredge] ($(C2)!0.3!(x2)$)
(C2) edge[minoredge] ($(C2)!0.25!(x3)$)
(C2) edge[minoredge] ($(C2)!0.2!(x4)$)
(Ck) edge[minoredge] ($(Ck)!0.2!(x2)$)
(Ck) edge[minoredge] ($(Ck)!0.3!(x3)$)
(Ck) edge[minoredge] ($(Ck)!0.3!(x4)$)
(x1) edge[] node[labelOnEdge,pos=0.5]{$b_1$} (t1)
(x1) edge[] node[labelOnEdge,pos=0.5]{$b_2$} (t2)
(x2) edge[minoredge] ($(x2)+(-0.4,-0.7)$)
(x2) edge[minoredge] ($(x2)+(0.4,-0.7)$)
(x3) edge[minoredge] ($(x3)+(-0.4,-0.7)$)
(x3) edge[minoredge] ($(x3)+(0.4,-0.7)$)
(x4) edge[minoredge] ($(x4)+(-0.4,-0.7)$)
(x4) edge[minoredge] ($(x4)+(0.4,-0.7)$)
;

\end{tikzpicture}
\vspace{3mm}
\caption{Reduction from {\betagapsat}. In this example, there are four variables $x_1,\dots, x_4$ in $\varphi$ and $C_1 = (x_1 \vee v_2 \vee \neg x_4)$ (the other clauses are omitted).
Hence, in the game instance, $C_1$ transitions to $x_1$ so long as agent~$1$ plays $a_1$ (irrespective of the actions played by agent~$2$ and the principal).
It transitions to $x_2$ if agent~$1$ plays $a_2$.
It transitions to $x_4$ if agent~$1$ plays $\aout$ and agent~$2$ plays $a_3$ (as $\neg x_4$ is the {\em third} literal of $C_1$).
The subtrees rooted at $x_2,\dots,x_4$ share the same structure as the one rooted at $x_1$.
}
\label{fig:reduction}
\end{figure}
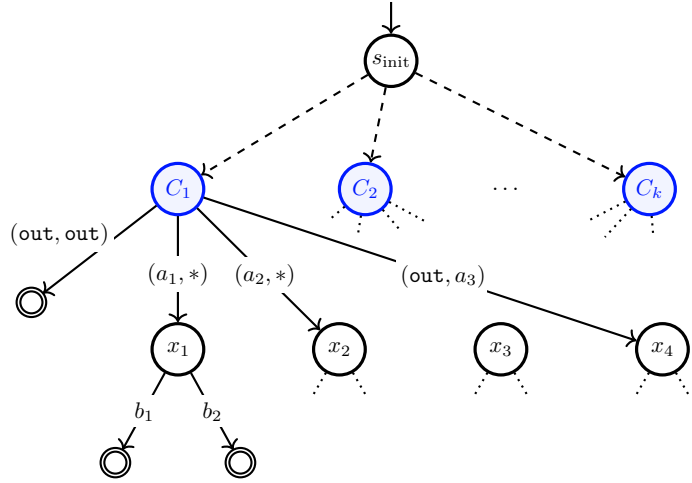

Consider a {\betagapsat} instance $\varphi$.
Let $x_1, \dots, x_m$ be the set of boolean variables in $\varphi$ and $C_1, \dots, C_k$ be the clauses. 
Each clause $C_j$ is a conjunction of three literals, each literal is either $x_i$ or its negation $\neg x_i$ for some $i \in \{1,\dots, m\}$.

We reduce $\varphi$ to a game with two agents and horizon length $3$.
As illustrated in \Cref{fig:reduction}, 
the state space of the game is $\{s_\init, s_\termin\} \cup \{c_1,\dots, c_k\} \cup \{x_1, \dots, x_m\}$.
Each agent has four actions $a_1, a_2, a_3$, and $\aout$, and the principal has two actions $b_1$ and $b_2$. 
\begin{itemize}
\item 
$s_\init$ is the starting state. Irrespective of the joint action performed, $s_\init$ transitions to each $C_j$ with probability $1/k$ and generates a reward of $0$.

\item Each state $C_j$ corresponds to the clause $C_j$ in $\varphi$.
W.l.o.g., we let each action $a_\ell$, $\ell \in \{1,2,3\}$, be available for agent~$1$ only if the $\ell$-th literal in $C_j$ is a positive literal; and let it be available for agent~$2$ only if the $\ell$-th literal in $C_j$ is a negative literal. The action $\aout$ is always available for both agents.

Irrespective of the principal's action, each $C_j$ transitions to:
    \begin{itemize}
    \item 
    $x_i$ with probability $1$, if $x_i$ is the $\ell$-th literal in $C_j$ and agent~$1$ plays $a_\ell$ (irrespective of agent~$2$'s action);
    \item 
    $x_i$ with probability $1$, if $\neg x_i$ is the $\ell$-th literal in $C_j$, agent~$1$ plays $\aout$, and agent~$2$ plays $a_\ell$;
    \item 
    $s_\termin$ with probability $1$, if both agents play $\aout$.
    \end{itemize}
Here, when both agents play $\aout$, each of them receives a reward of $\rmax$ and the principal receives $0$; otherwise, every players receives $0$.
(Recall that $\rmax = \frac{1-\gamma}{2}$ and $\rmin = - \frac{1- \gamma}{2}$.)

\item 
Each state $x_i$ corresponds to the variable $x_i$ in $\varphi$.
The principal effectively controls this state; all rewards depend only on the principal's action:
    \begin{itemize}
    \item For each $i\in \{1,2\}$, when $b_i$ is played, agent~$i$ gets $\rmax$ and the other agent gets $\rmin$.
    \end{itemize}
In all these cases, the principal gets a reward of $\rmax$, and the state transitions to $s_\termin$.

\item The game terminates at $s_\termin$.
\end{itemize}

Let $\delta = \gamma \cdot \rmax/2$.
We show that if the {\betagapsat} instance is a yes-instance, then there exists a stationary (exact) CE inducing value $\gamma^2$ for the principal.
Otherwise, no stationary $\delta$-CE induces value higher than $\gamma^2 \beta$. Hence, taking $\epsilon = \gamma^2 \beta$ gives the stated result.
Note that it is essential that the yes-instance corresponds to an exact CE because of the resource augmentation in the approximation criterion considered. 

Intuitively, the principal can freely adjust the agents' values at each $x_i$ within the range $[\rmin,\rmax]$, so long as the values sum to zero.
This associates actions $b_1$ and $b_2$ with the $true$ and $false$ assignments to $x_i$ in the SAT instance:
\begin{itemize}
\item Playing $b_1$ with probability $1$ at $x_i$ yields $\rmax$ for agent~$1$. This is sufficiently high to incentivize the agent to play $a_\ell$ at every clause where $x_i$ appears as the $\ell$-th literal, just as letting $x_i = true$ makes all these clauses true in the SAT instance.

\item Similarly, playing $b_2$ with probability $1$ incentivizes agent~$2$ to play $a_\ell$ at clauses where $\neg x_i$ appears as the $\ell$-th literal, just as letting $x_i = false$ makes all these clauses true in the SAT instance.
\end{itemize}
More precisely, in the case of $\delta$-CEs, any reward as high as $\rmax - \delta/\gamma$ is sufficient to provide the desired incentive. 
The parameter $\delta = \gamma \cdot \rmax/2$ further ensures that {\em only one} of the two agents can be incentivized at each $x_i$, just as $x_i$ can be assigned only one value in the SAT instance.

Since the principal obtains zero reward everywhere except at the $x_i$ states, her overall value is proportional to the number of clause states where at least one agent opts for an action other than $\aout$, so that the game proceeds to a variable state.
This establishes a connection between computing a maximum value CE and finding a truth assignment satisfying the maximum number of clauses in $\varphi$.
More formally, let us analyze the cases where $\varphi$ is a yes-instance and a no-instance, respectively.

If $\varphi$ is a yes-instance, there exists a satisfiable assignment $f: \{x_1,\dots,x_m\} \to \{true, false\}$.
One can verify that playing $b_1$ at every $x_i$ such that $f(x_i) = true$, and $b_2$ at every $x_i$ such that $f(x_i) = false$ disincentivizes at least one agent from playing $\aout$. 
Hence, with probability $1$, the game will reach some state $x_i$, at which the principal receives a reward of $1$. The discounted sum of the principal's reward is $\gamma^2$ in this case.

Conversely, suppose that $\varphi$ is a no-instance and consider an arbitrary stationary $\delta$-CE $\pi$ in the corresponding game instance.
The stationarity of $\pi$ means that it employs a fixed distribution at each $x_i$, independent of the history.
Since $\varphi$ is a no-instance, no assignment satisfies more than $\beta k$ clauses.
This implies that, in the game instance, for any choice of action distributions at the $x_i$'s, there are more than $\beta k$ clause states where it is {\em not} $\delta$-optimal for any agent to play any action other than $\aout$.  
As a result, under $\pi$, the game terminates at $s_\termin$ with probability at least $1-\beta$, without going through any $x_i$. Hence, $\pi$ induces a value of at most $\gamma^2 \beta$. This completes the proof.
\end{proof}

\subsection{Proofs in \Cref{sc:value-set-iteration}}

\prpVstar*

\begin{proof}
To prove the ``only if'' direction, suppose that $\vv \in \calV^\star(s)$ and we demonstrate the existence of $\bar{\pi}$ and $\ww$ satisfying the stated conditions.

By \Cref{def:inducible-value-set}, $\vv \in \calV^\star(s)$ means that it can be induced by some CE $\pi$, such that $\vv = V^\pi(s)$.
It can be verified that letting $\bar{\pi} = \pi(s)$ and $\ww(\va, \va; s') = V^{\pi}(s, \va, \va; s')$ for all $\va \in \bfA$ and $s' \in S$ gives the stated conditions.

Specifically,
$\vv = \bellman(s, \bar{\pi}, \ww)$ follows readily from \Cref{eq:Q,eq:V}.
Then, \Cref{eq:const-ic} follows from \Cref{eq:IC-Q}, considering deviation plans in which agent $i$ deviates to action $b$ in the first time step and then follows $\pi$ throughout.
Finally, \Cref{eq:const-onward} holds because the subsequent execution of $\pi$ given each sequence $(s,\va,\bb;s')$ remains a CE in the subgame starting at $s'$, and $\ww(\va,\bb;s')$ is precisely the value it induces in the subgame. Hence, $\ww(\va,\bb;s') \in \calV^\star(s')$ according to the definition of $\calV^\star$.

Now consider the ``if'' direction.
Suppose that $\vv = \bellman(s, \bar{\pi}, \ww)$, and  $(\bar{\pi}, \ww)$ satisfies \Cref{eq:const-ic,eq:const-onward}.
The fact that $\ww(\va,\bb;s') \in \calV^\star(s')$ for all $\va, \bb, s'$ implies that each $\ww(\va,\bb;s')$ can be induced by some CE, say $\pi'_{\va,\bb,s'} : \Sigma \times S \to \Delta(\bfA)$.
We construct the following policy $\pi$, which concatenates these CEs with $\bar{\pi}$: 
\begin{itemize}
\item 
let $\pi(s) = \bar{\pi}(s)$; and

\item
let 
$\pi(s, \va,\bb;\sigma; s') = \pi'_{\va,\bb,s^1}(\sigma;s')$ for all $\va,\bb,s'$, and $\sigma = (s^1,\va^1,\bb^1;\dots;s^\ell, \va^\ell, \bb^\ell) \in \Sigma$.
\end{itemize}
(The other unspecified distributions in $\pi$ can be arbitrary---they are outputs of $\pi$ on input sequences starting at some $s'\neq s$.)
It is then follows from \Cref{eq:const-ic} that $\pi$ forms a CE and from $\vv = \bellman(s, \bar{\pi}, \ww)$ that $\vv = V^\pi(s)$. Hence, $\vv \in \calV^\star(s)$.
\end{proof}

\lmmPhiproperties* 

\begin{proof}
We prove the properties below.

\paragraph{Convexity}
Consider arbitrary $\vv, \vv' \in \Phi(\calV)(s)$.
Suppose that 
$\vv = \bellman(s, \bar{\pi}, \ww)$ and
$\vv' = \bellman(s, \bar{\pi}', \ww')$,
where $(\bar{\pi}, \ww), (\bar{\pi}', \ww') \in \calF_s(\calV)$.
Consider executing one of $(\bar{\pi},\ww)$ and $(\bar{\pi}',\ww')$ uniformly at random, the effect of which is equivalent to the tuple $(\bar{\pi}'',\ww'')$ where:
\begin{itemize}
\item 
$\bar{\pi}''(\va) = (\bar{\pi}(\va) + \bar{\pi}'(\va))/2$ for every $\va \in \bfA$; and

\item 
$\ww''(\va, \bb; s') = \frac{\bar{\pi}(\va)}{\bar{\pi}(\va) + \bar{\pi}'(\va)} \cdot \ww(\va, \bb; s') + \frac{\bar{\pi}'(\va)}{\bar{\pi}(\va) + \bar{\pi}'(\va)} \cdot \ww'(\va, \bb; s')$ for every $(\va, \bb; s') \in \bfA^2 \times S$.
\end{itemize}
By linearity, we get that $\vv'' \coloneqq (\vv + \vv')/2 =  \bellman(s, \bar{\pi}'', \ww'')$ and \Cref{eq:const-ic} holds.
Moreover, $\bar{\pi}'' \in \Delta(\bfA)$ and $\ww''(\cdot,\cdot;s') \in \calV(s')$ for all $s'$ due to the convexity of $\Delta(\bfA)$ and $\calV(s')$.

Consequently, $(\bar{\pi}'',\ww'') \in \calF_s(\calV)$ and $\vv'' \in \Phi(\calV)$. Since $\vv$ and $\vv'$ are chosen arbitrarily, $\Phi(\calV)(s)$ is convex.

\paragraph{Closedness}
To see the closedness of $\Phi(\calV)(s)$.
Consider any convergent sequence $\vv^1, \vv^2, \dots \in \Phi(\calV)(s)$, and let $\vv$ be the limit point. The task is to argue that $\vv \in \Phi(\calV)(s)$.

Since $\vv^k \in \Phi(\calV)(s)$, by definition, 
$\vv^k = \bellman(s, \bar{\pi}^k, \ww^k)$ for some $(\bar{\pi}^k, \ww^k) \in \calF_s(\calV)$.
Consider the sequence $\{(\vv^k, \bar{\pi}^k, \ww^k)\}_{k=1}^\infty$. It is bounded, so according to the Bolzano–Weierstrass theorem it must have a convergent subsequence, that converges to some $(\vv, \bar{\pi}, \ww)$. 
Note that $\calF_s(\calV)$ is closed if $\calV(s')$ is closed for every $s'$.
Hence, it must be that $(\bar{\pi}, \ww) \in \calF_s(\calV)$.
Moreover, by continuity, now that $\vv^k = \bellman(s, \bar{\pi}^k, \ww^k)$ for every $k$, this equation must hold at the limit point too, so $\vv = \bellman(s, \bar{\pi}, \ww)$.
By definition, this means that $\vv \in \Phi(\calV)(s)$.

\paragraph{Nonemptiness}
Suppose that we pick the onward values $\ww$ in a way such that they are invariant with respect to the actions recommended and played, i.e., $\ww(\va,\bb; s')$ depends only on $s'$ and we can write it as $\ww(s')$. 
Then, any $\bar{\pi} \in \Delta(\bfA)$ that satisfies \Cref{eq:const-ic} is effectively a CE of the one-shot game where the utility generated by each joint action $\va \in \bfA$ is $u_i(\bb) = r_i(s,\va) + \gamma \mathbb{E}_{s'\sim p(\cdot \given s, \va)} \ww(s')$ for each agent $i$.
A CE always exists in every one-shot game, and we have $\bellman(s, \bar{\pi}, \ww) \in \Phi(\calV)(s)$, so $\Phi(\calV)(s)$ is nonempty.
\end{proof}

\lmmPhifixedpoint*

\begin{proof}
The lemma follows from the same proof for the ``achievability'' property shown by \citet{murray2007finding} (see Lemma~3 in their work).
For a complete argument, note that \Cref{prp:eps-neigh-inducible} is a more general version of this lemma.
Choosing $\epsilon = \delta = 0$ there gives stated result. 
\end{proof}

\lmmPhimonotonicity*

\begin{proof}
This can be verified by noting that expanding the value sets in the input to $\Phi$ effectively relaxes the onward value constraint in \Cref{eq:const-onward}.
\end{proof}

\lmmnestedsets*

\begin{proof}
Recall that all rewards in $\rr$ are bounded in $[\rmin, \rmax] = [-\frac{1-\gamma}{2},\, \frac{1-\gamma}{2}]$.
Consider \Cref{eq:const-bellman}. 
If the onward vectors are chosen from $\calB = [-\frac{1}{2}, \frac{1}{2}]^{n+1}$, we can easily bound $\vv = \bellman(s, \bar{\pi}, \ww)$ within the space $[\rmin - \gamma/2,\, \rmax + \gamma/2]^{n+1} = \calB$.
This means that $\Phi(\calB) \subseteq \calB$.

Now that $\calV^0 = \calB$, it holds automatically that $\calV^1 \subseteq \calV^0$ and $\calV^\star \subseteq \calV^0$.
By monotonicity (\Cref{lmm:Phi-monotonicity}), applying $\Phi$ to the two sides of $\calV^\star \subseteq \calV^0$ then gives
\begin{align*}
\calV^\star = \Phi(\calV^\star) \subseteq \Phi(\calV^0) = \calV^1.
\end{align*}
Hence, $\calV^\star \subseteq \calV^1 \subseteq \calV^0$.
By repeatedly applying $\Phi$ to this equation and using the monotonicity $\Phi$, we get that $\calV^\star \subseteq \calV^{k+1} \subseteq \calV^k$ for all $k \in \mathbb{N}$, as desired.
\end{proof}

\thmPhiconverge*

\begin{proof}
Given \Cref{lmm:nested-sets}, it is sufficient to argue that 
$\calV^\infty \subseteq \calV^\star$ in order to prove the convergence of the sequence to $\calV^\star$, where 
\[
\calV^\infty \coloneqq \bigcap_{k=1}^\infty \calV^k.
\]
We show this next.

Consider arbitrary $s\in S$ and $\vv \in \calV^\infty(s)$, which means $\vv \in \calV^k(s)$ for all $k \in \mathbb{N}$. 
Since $\calV^{k+1} = \Phi(\calV^k)$, we also have $\vv \in \Phi(\calV^k)(s)$. 
By definition, this means that, for every $k \in \mathbb{N}$, 
\[
\vv = \bellman(s,\bar{\pi}^k, \ww^k)
\]
for some $(\bar{\pi}^k, \ww^k) \in \calF_s(\calV^k)$.
Consider the sequence $\{(\bar{\pi}^k, \ww^k)\}_{k=K}^\infty$.
Since $\bar{\pi}^k$ and $\ww^k$ are both bounded, by the Bolzano–Weierstrass theorem,
there must be a convergent subsequence $\{(\bar{\pi}^{k_\ell}, \ww^{k_\ell})\}_{\ell=1}^\infty$. Let $(\bar{\pi}, \ww)$ be the limit of this subsequence.
We claim that the following holds.

\begin{claim}
\label{clm:1}
$(\bar{\pi}, \ww) \in \calF_s(\calV^\infty)$.
\end{claim}

\begin{proof}
By definition, we need to prove that $\bar{\pi} \in \Delta(\bfA)$, $\ww(\va, \bb; s') \in \calV^\infty(s')$ for all $\va$, $\bb$, $s'$, and $(\bar{\pi}, \ww)$ satisfies \Cref{eq:const-ic}. 

First, $\bar{\pi} \in \Delta(\bfA)$ because $\Delta(\bfA)$ is closed and $\bar{\pi}$ is the limit of $\{\bar{\pi}^{k_\ell}\}_{\ell=1}^\infty$, which is contained in $\Delta(\bfA)$. 
Indeed, each $\calV^k(s')$ is also closed, as implied by \Cref{lmm:Phi-properties} and the fact that every set in $\calV^0$ is closed.
Given this, we can argue that $\ww(\va, \bb; s') \in \calV^\infty(s')$ for any $\va$, $\bb$, $s'$.

Specifically, if this were not true, then by definition $\ww(\va, \bb; s') \notin \calV^K(s')$ for some $K$. 
On the other hand, for $L$ sufficiently large, we must have $\ww^{k_\ell}(\va, \bb; s') \in \calV^{k_\ell}(s') \subseteq \calV^{K}(s')$ for all $\ell \ge L$. 
Now that $\calV^K(s')$ is closed, the subsequence $\{\ww^{k_\ell}(\va,\bb;s')\}_{\ell=L}^\infty$ must converge to a point inside $\calV^K(s')$, so we get that $\ww(\va,\bb;s') \in \calV^{K}(s')$, which is a contradiction.

Finally, for every $k \in \mathbb{N}$, $(\bar{\pi}^k, \ww^k)$ satisfies \Cref{eq:const-ic} because $(\bar{\pi}^k, \ww^k) \in \calF_s(\calV^k)$.
By continuity of the two sides of \Cref{eq:const-ic}, the constraint must hold for the limit $(\bar{\pi}, \ww)$, too. 
\end{proof}

Now, since $\vv = \bellman(s, \bar{\pi}^k, \ww^k)$ for every $k$, by continuity, we have $\vv = \bellman(s, \bar{\pi}, \ww)$, too.
Combining this with \Cref{clm:1} then gives $\vv \in \Phi(\calV^\infty)(s)$.
As this holds for every $s$ and every $\vv \in \calV^\infty(s)$, it follows that
$\calV^\infty = \Phi(\calV^\infty)$, so $\calV^\infty$ is a fixed point of $\Phi$.
Recall that $\calV^\star$ is the greatest fixed point of $\Phi$ (\Cref{prp:Phi-fixed-point}), which means $\calV^\infty \subseteq \calV^\star$, as desired.

Finally, $\calV^\infty \neq \emptyset$ follows from the fact that every $\calV^k(s)$, $s \in S$, is closed and nonempty: every sequence $\vv^1,\vv^2,\dots$ such that $\vv^k \in \calV^k(s)$ must have a convergent subsequence, and similarly to an argument we used above, the limit point of this subsequence must be inside $\calV^\infty(s)$.
\end{proof}

\lmmxPhiPhi*

\begin{proof}
We first prove $\xPhi(\calV) \subseteq \neigh_\xi(\Phi(\calV))$.
By construction, $\xPhi(\calV)(s)$ is the convex hull of a subset of $\neigh_\xi(\Phi(\calV)(s))$.
Now that $\calV(s)$ is convex for all $s$, by \Cref{lmm:Phi-properties}, $\Phi(\calV)(s)$ and, consequently, $\neigh_\xi(\Phi(\calV)(s))$ are also convex.
A convex set contains the convex hull of every subset of itself, so we have $\xPhi(\calV) \subseteq \neigh_\xi(\Phi(\calV))$.

To see that $\Phi(\calV) \subseteq \xPhi(\calV)$, note that the grid cell containing any $\vv \in \Phi(\calV)(s)$ is contained entirely in $\neigh_\xi(\Phi(\calV)(s))$.
Hence, all the vertices of the cell are included in $\xPhi(\calV)(s)$. Consequently, $\vv \in \xPhi(\calV)(s)$.
\end{proof}

\thmepsneighinducible*

\begin{proof}
This follows readily by \Cref{lmm:eps-neigh-inducible-alg}: the policy $\pi$ computed by \Cref{alg:compute-pi} is a $\delta$-CE such that $V^\pi(s_\init) \in \neigh_\epsilon(\vv)$. Hence, $\vv$ is $(\epsilon,\delta)$-inducible at $s_\init$.
\end{proof}

\lmmepsneighinduciblealg*

\begin{proof}
First, let us verify that $\pi$ is well-defined, i.e., for any given sequence $(\sigma;s)$ starting at $s_\init$, \Cref{alg:compute-pi} successfully outputs a distribution. 
The key is to argue that, in every iteration $t$, there exists $(\bar{\pi}, \ww) \in \calF_{s^t}(\calV)$ that implements some $\tilde{\vv}^t \in \neigh_\xi(\vv^t)$ at \Cref{ln:compute-vv}. Equivalently, this means that $\neigh_\xi(\vv^t) \cap \Phi(\calV)(s^t) \neq \emptyset$ since, by definition, $\Phi(\calV)(s^t)$ consists of all such implementable values.

Consider first the case where $t=1$.
According to the input requirement, $\vv^1 = \vv \in \calV(s^1)$ and $\calV(s^1) \subseteq \neigh_\xi(\Phi(\calV)(s^1))$. 
Hence, $\vv^1$ lies in the $\xi$-neighborhood of $\Phi(\calV)(s^1)$, which also means 
$\neigh_\xi(\vv^1) \cap \Phi(\calV)(s^1) \neq \emptyset$, as desired.
The fact that $(\bar{\pi},\ww) \in \calF_{s^1}(\calV)$ further ensures that 
$\vv^2 = \ww(\va^1,\bb^1;s^2) \in \calV(s^2)$, so the same argument applies to $t=2$.
By induction, we get that $\neigh_\xi(\vv^t) \cap \Phi(\calV)(s^t) \neq \emptyset$ for all $t$.
Hence, $\pi$ is well-defined.
\smallskip

We next show that $\pi$ is a $\delta$-CE and $V^\pi(s_\init) \in \neigh_\epsilon(\vv)$ to complete the proof.
For better clarity, let us define a function $\uu: \Sigma \times S \to \mathbb{R}^{n+1}$ such that, for every sequence $(\sigma;s) \in \Sigma \times S$, $|\sigma| = \ell$,
\begin{align}
\label{eq:uu}
\uu(\sigma;s) = \vv^{\ell+1}(\sigma;s),
\end{align}
where $\vv^{\ell+1}(\sigma;s)$ denotes the value of $\vv^{\ell+1}$ in \Cref{alg:compute-pi}, given input $(\sigma;s)$.
Similarly, we define a reward scheme $\tilde{\rr}: \Sigma \times S \to \mathbb{R}^{n+1}$,
such that
\begin{align}
\label{eq:td-rr}
\tilde{\rr}(\sigma; s) = \vv^{\ell+1}(\sigma;s) - \tilde{\vv}^{\ell+1}(\sigma;s).
\end{align}

Consider executing $\pi$ in the following games.
We will compare the V- and Q-values generated in these games to argue that: in the original game no agent benefits by more than $\delta$ by deviating from $\pi$, while $\pi$ induces some $\tilde{\vv}$ close to $\vv$.
\begin{itemize}
\item[1.] 
{\bf Subsidized game $\Gsub$}: During the execution of $\pi$, each player is subsidized by $\tilde{\rr}$ according to the interaction history. 
Every time a new state $s$ is generated following history $\sigma$, every player~$i$ receives a reward $\tilde{r}_i(\sigma;s)$, in addition to their original rewards given by $\rr$. 

\item[2.] 
{\bf Truncated game $\Gtrun$}: The game is subsidized as above. Additionally, we terminate it at time step $t+1$. 
In this final time step, regardless of the joint action performed, every player~$i$ receives only a reward $u_i(\sigma;s^{t+1})$ based on the history $(\sigma; s^{t+1})$.

\item[3.] 
{\bf The Original game $\calG$} (where neither $\tilde{\rr}$ nor $\uu$ is applied).  
\end{itemize}

\paragraph{Subsidized vs. Truncated}
We first compare $\Gsub$ and $\Gtrun$ and show that $\pi$ forms a CE in both games and induces $\vv$.
In what follows, we denote by $V^{\pi,\tilde{\rr}}$ and $\overline{V}^{\pi,\tilde{\rr}}$ the V-values induced by $\pi$ in $\Gsub$ and $\Gtrun$, respectively. The same notation applies to the Q-values.

For $\Gtrun$, this can be verified by induction.
\begin{itemize}
\item 
As the base case, $\pi$ is automatically IC at time step $t+1$, where $\Gtrun$ terminates: 
for any sequence $(\sigma;s,a)$, $|\sigma|=t$, \Cref{eq:IC-Q} holds.
Moreover,
\begin{align}
\label{eq:V-trunc}
\overline{V}^{\pi,\tilde{\rr}}(\sigma;s) = \uu(\sigma; s).
\end{align}

\item
Consider a generic $\ell \le t$.
Suppose that the above conditions hold for time step $\ell + 1$. We show that they must hold for step $\ell$, too.
Specifically, consider arbitrary $(\sigma;s)$, $|\sigma|=\ell$.
\Cref{eq:V-trunc} can be verified by establishing the following: 
\begin{align*}
&\overline{V}^{\pi,\tilde{\rr}}(\sigma; s) \\
&= 
\E_{\va \sim \pi(\sigma;s)} \Big( \rr(s, \va) + \tilde{\rr}(\sigma;s) + \gamma \E_{s' \sim p(\cdot \given s, \va)} \overline{V}^{\pi,\tilde{\rr}}(\sigma;s,\va,\va;s') \Big) 
& \text{\small (by \Cref{eq:V,eq:Q})}\\
&= 
\E_{\va \sim \pi(\sigma;s)} \Big( \rr(s, \va) + \tilde{\rr}(\sigma;s) + \gamma \E_{s' \sim p(\cdot \given s, \va)} \uu(\sigma;s,\va,\va;s') \Big) & \text{\small (by \Cref{eq:V-trunc} for $\ell+1$)} \\
&= 
\E_{\va \sim \pi(\sigma;s)} \Big( \rr(s, \va) + \tilde{\rr}(\sigma;s) + \gamma \E_{s' \sim p(\cdot \given s, \va)} \vv^{\ell+2}(\sigma;s,\va,\va;s') \Big) & \text{\small (by \Cref{eq:uu})}\\
&= 
\E_{\va \sim \bar{\pi}} \Big( \rr(s, \va) + \gamma \E_{s' \sim p(\cdot \given s, \va)} \ww(\va,\va;s') \Big) + \tilde{\rr}(\sigma;s) & \text{\small (by \Cref{ln:vv-t}, Alg.~\ref{alg:compute-pi})} \\
&=
\bellman(s, \bar{\pi}, \ww) + \tilde{\rr}(\sigma;s)  \\
&= 
\tilde{\vv}^{\ell+1}(\sigma;s) +  \tilde{\rr}(\sigma;s) & \text{\small (by \Cref{ln:compute-vv}, Alg.~\ref{alg:compute-pi})} \\
&= \uu(\sigma; s), & \text{\small (by \Cref{eq:td-rr,eq:uu})}
\end{align*}
where $\bar{\pi}$ and $\ww$ denote values of these variables at the $(\ell+1)$-th iteration of \Cref{alg:compute-pi}, given input $(\sigma;s)$.
Note that the values of these variables are the same for the input sequences $(\sigma;s)$ and $(\sigma;s,\va,\va;s')$, when evaluated at the same iteration.

Additionally, \Cref{eq:IC-Q} can be verified by expanding $\overline{Q}_i^{\pi,\tilde{\rr}}(\sigma;s,\va)$ and $\overline{Q}_i^{\pi,\rho,\tilde{\rr}}(\sigma;s,\va)$ in a similar way, and noting that $(\bar{\pi}, \ww) \in \calF_{s^{\ell+1}}(\calV)$ implies the IC condition in \Cref{eq:const-ic}.
Similarly to our argument in the proof of \Cref{prp:V-star}, this IC condition implies \Cref{eq:IC-Q}.  
\end{itemize}
Hence, by induction, \Cref{eq:IC-Q} holds throughout, $\pi$ is a CE in $\Gtrun$, and $\overline{V}^{\pi,\tilde{\rr}}(s_\init) =  \uu(s_\init) = \vv^1(s_\init) = \vv$.

Now we turn to $\Gsub$.
The differences between the following values can be bounded by $\gamma^t$, for any $s,\va$:
\begin{align*}
\left\| V^{\pi,\rho, \tilde{\rr}}(s_\init) - \overline{V}^{\pi,\rho, \tilde{\rr}}(s_\init) \right\|_\infty \le \gamma^t  
\quad\text{ and }\quad
\left\| Q^{\pi,\rho, \tilde{\rr}}(s_\init, \va) - \overline{Q}^{\pi,\rho, \tilde{\rr}}(s_\init, \va) \right\|_\infty \le \gamma^t.
\end{align*}
Indeed, in the first $t$ time steps, rewards generated along every sequence are the same in both games.
As for rewards in the subsequent steps, even if they differ arbitrarily in $\Gsub$ and $\Gtrun$, the discounted sum of the differences is at most $\sum_{\ell = t+1}^\infty \gamma^{\ell-1} \cdot \rmax \le \gamma^t$. Hence, the above bounds follow.

Besides $s_\init$, the same analysis also applies to any other sequences $(\sigma,s)$, $|\sigma|= \ell$, whereby we get: 
\begin{align*}
\left\| V^{\pi,\rho, \tilde{\rr}}(\sigma; s) - \overline{V}^{\pi,\rho, \tilde{\rr}}(\sigma; s) \right\|_\infty \le \gamma^{t+1-\ell}  
\quad\text{ and }\quad
\left\| Q^{\pi,\rho, \tilde{\rr}}(\sigma; s, \va) - \overline{Q}^{\pi,\rho, \tilde{\rr}}(\sigma; s, \va) \right\|_\infty \le \gamma^{t+1-\ell}.
\end{align*}
Taking $t \to \infty$, we get that the V- and Q-values in $\Gsub$ are the same for every possible sequence as those in $\Gtrun$.
Since $\pi$ forms a CE in $\Gtrun$, this means that it must be a CE in $\Gsub$, too.
In particular, $V^{\pi, \tilde{\rr}}(s_\init) = \overline{V}^{\pi, \tilde{\rr}}(s_\init) = \vv$.

\paragraph{Subsidized vs. Original}

Next, we compare $\Gsub$ and the original game $\calG$ to bound the benefit of deviating from $\pi$ in the latter.
We can establish the following bounds for any $\sigma,s,\va$, by noting that removing the subsidies in $\Gsub$ results in a difference of at most $\max_{(\sigma,s) \in \Sigma\times S}|\tilde{\rr}(\sigma,s)|$ in every time step:
\begin{align}
\label{eq:G-Gsub-values}
\left\| V^{\pi,\rho}(\sigma;s) - V^{\pi,\rho, \tilde{\rr}}(\sigma;s) \right\|_\infty \le \epsilon
\quad\text{ and }\quad
\left\| Q^{\pi,\rho}(\sigma;s, \va) - Q^{\pi,\rho, \tilde{\rr}}(\sigma; s, \va) \right\|_\infty \le \delta/2,
\end{align}
where we used the fact that 
$\max_{(\sigma,s) \in \Sigma\times S}|\tilde{\rr}(\sigma,s)| \le \xi$ (since in \Cref{alg:compute-pi} every $\tilde{\vv}^t$ is in the $\xi$-neighborhood of $\vv^t$) and 
$\sum_{t=1}^\infty \gamma^{t} \cdot \xi \le \xi/(1-\gamma) = \min\{\epsilon, \delta/2\}$.

As demonstrated, $\pi$ forms a CE in $\Gsub$. It then follows that in the original game the benefit from deviating from $\pi$ is at most $\delta$ (twice the upper bound of the Q-value difference above).
Moreover, in the original game,
\[
\left\| V^{\pi}(s_\init) - \vv \right\|_\infty \le 
\left\| V^{\pi, \tilde{\rr}}(s_\init) - \vv \right\|_\infty + \epsilon = \epsilon,
\]
where we used $V^{\pi, \tilde{\rr}}(s_\init) = \vv$, as demonstrated above.

In summary, $\pi$ forms a CE in the original game $\calG$ and $V^\pi(s_\init) \in \neigh_\epsilon(\vv)$. 
This completes the proof.
\end{proof}

\prpalgiterationoutput* 

\begin{proof}
We first bound the number of iterations.
Consider the sequence of value-set functions $\xcalV^0, \xcalV^1, \dots$ generated by \Cref{alg:value-set-iteration}, and let $G^k(s)$ be the set of grid points inside $\xcalV^k(s)$.
Observe that $\xPhi$ has the same monotonicity property as $\Phi$, as stated in \Cref{lmm:Phi-monotonicity}. 
Hence, by a similar argument as that in the proof of \Cref{thm:Phi-converge}, $\xPhi$ maps $\neigh_\xi(\calB)$ to a subset, and initializing $\xcalV^0$ to $\neigh_\xi(\calB)$ ensures $\xcalV^{k+1} \subseteq \xcalV^k$ for all $k \in \mathbb{N}$. 
Consequently, $G^{k+1} \subseteq G^k$. 

By construction, each $\xcalV^k(s)$ is the convex hull of $G^k(s)$. So, $\xcalV^{k+1} = \xcalV^k$ if and only if $G^{k+1} = G^k$.
As the algorithm continues, at least one grid point must be removed from $G^k(s)$ in each iteration, for some $s \in S$. The algorithm then terminates in at most $|S| \cdot (1/\xi)^{n+1}$ iterations, before all grid points are removed from the $G^k(s)$'s.

It remains to prove the stated properties of $\calV$.
First, $\calV = \xPhi(\calV)$ is evident from the termination condition.
To see that $\calV \supseteq \calV^\star$, recall that $\Phi(\calV) \subseteq \xPhi(\calV)$ by \Cref{lmm:xPhi-Phi}.
This implies that $\calV^k \subseteq \xcalV^k$ for all $k \in \mathbb{N}$
(where $\calV^k$ is the $k$-th value-set function generated by the algorithm under $\Phi$).
Further recall that $\calV^k \supseteq \calV^\star$ by \Cref{lmm:nested-sets}, so $\xcalV^k \supseteq \calV^\star$ for all $k$, and  $\calV \supseteq \calV^\star$ then follows.
\end{proof}

\lmmcomputeinducibility*

\begin{proof}
Note that $\vv \in \neigh_\beta(\Phi(\calV)(s))$ can be stated as: there exists $\tilde{\vv} \in \neigh_\beta(\vv) \cap \Phi(\calV)(s)$.
We formulate this decision problem as a constraint satisfiability problem: deciding whether there exist $\tilde{\vv} \in \mathbb{R}^{n+1}$, $\bar{\pi} \in \Delta(\bfA)$, and $\yy(\va,\bb,s') \in \mathbb{R}^{n+1}$ for $(\va, \bb, s') \in \bfA^2 \times S$ such that the constraints listed below hold. Here, each $\yy(\va,\bb,s')$ corresponds to the onward value $\ww(\va,\bb,s)$, expressed as a convex combination of the vertices of $\calV(s')$, i.e., $\ww(\va,\bb,s') = M(s') \cdot \yy(\va,\bb,s')$, where $M(s')$ is a matrix whose columns correspond to the vertices of $\calV(s')$ and we view $\yy(\va,\bb,s')$ as a column vector.
\begin{itemize}
\item Neighborhood constraint, ensuring $\tilde{\vv} \in \neigh_\beta(\vv)$:
\[
\vv - \beta \le \tilde{\vv} \le \vv + \beta.
\]

\item Bellman constraint, as \Cref{eq:const-bellman} but with $\tilde{\vv}$ being the value implemented:
\begin{align*}
\tilde{\vv} =  \sum_{\va \in \bfA} \bar{\pi}(\va) \cdot \left( \rr(s, \va) + \gamma \sum_{s' \in S} p(s' \given s, \va) \cdot \yy(\va, \va, s') \right).
\end{align*}

\item IC constraint, as \Cref{eq:const-ic}, for every $a,b \in A$ and $i = 1, \dots, n$: 
\begin{align*}
&\sum_{\va \in \bfA: a_i = a} \bar{\pi}( \va) \cdot \left( r_i\left(s, \va \right) + \gamma \sum_{s' \in S} p\left(s' \given s, \va \right) \cdot 
M_i(s') \cdot \yy \left(\va, \va, s' \right) \right) \ge \nonumber\hspace{-30mm}\\
&
\qquad \sum_{\va \in \bfA: a_i = a} \bar{\pi}( \va) \cdot \left( r_i\left(s, \va \oplus_i b \right) + \gamma \sum_{s' \in S} p\left(s' \given s, \va \oplus_i b \right) \cdot M_i(s') \cdot \yy \left(\va, \va \oplus_i b, s' \right) \right),
\end{align*}
where $M_i(s')$ denotes the $i$-th row of $M_i(s')$ (so $M_i(s') \cdot \yy (\va, \va \oplus_i b, s') = w_i (\va, \va \oplus_i b, s')$).

\item Onward value constraint, as \Cref{eq:const-onward}.
Now that the onward values are encoded as $\yy$, $\ww(\va,\bb,s') \in \calV(s')$ is equivalent to saying that $\yy(\va, \bb, s')$ is a valid distribution over the columns of $M_i(s')$.
Hence, the following constraints are imposed, for every $\va,\bb \in \bfA$ and $s' \in S$:
\begin{align}
\label{eq:onward-y}
\mathbf{1} \cdot \yy(\va, \bb, s') = 1
\quad\text{and}\quad 
\yy(\va, \bb, s') \ge \mathbf{0}.
\end{align}
\end{itemize}

Since both $\bar{\pi}$ and $\yy$ are variables, the above constraints involve quadratic terms $\bar{\pi}(\va) \cdot \yy(\va, \bb, s')$. Nevertheless, these terms can be eliminated by using a standard approach in previous works \cite{macdermed2011quick,gan2023sequentialv2}.
Specifically, we introduce an auxiliary variable $\zz(\va, \bb, s')$ to replace each $\bar{\pi}(\va) \cdot \yy(\va, \bb, s')$. 
Additionally, we impose the following constraints on $\zz$ to replace \Cref{eq:onward-y}:
\begin{align}
\label{eq:const-auxiliary}
\mathbf{1} \cdot \zz(\va, \bb, s') = \bar{\pi}(\va)
\quad\text{and}\quad
\zz(\va, \bb, s') \ge \mathbf{0}.
\end{align}
Note that the $\yy$ variables do not appear in the new formulation, as they always appear in the quadratic terms in the original. 

Clearly, for any feasible solution $(\tilde{\vv}, \bar{\pi}, \yy)$ to the original formulation, $(\tilde{\vv}, \bar{\pi}, \zz)$, where $\zz(\va, \bb, s') = \bar{\pi}(\va) \cdot \yy(\va, \bb, s')$ for every $\va,\bb,s'$, constitutes a feasible solution to the new formulation.
Conversely, for any feasible solution $(\tilde{\vv}, \bar{\pi}, \zz)$ to the new formulation, $(\tilde{\vv}, \bar{\pi}, \yy)$ constitutes a feasible solution to the original formulation, where for each $\va, \bb, s'$:
$\yy(\va, \bb, s') = \zz(\va, \bb, s') / \bar{\pi}(\va)$ if $\bar{\pi}(\va) > 0$; and
$\yy(\va, \bb, s')$ is an arbitrary distribution if $\bar{\pi}(\va) = 0$.

Hence, the new formulation preserves the satisfiability of the original one, as well as the values of $\tilde{\vv}$ in the feasible solutions.
Moreover, the new formulation involves only linear constraints, so it can be solved as an LP in polynomial time.
When the constraints are satisfiable, we obtain a feasible solution including: a distribution $\bar{\pi}$ and onward values $\ww = M \cdot \yy$ such that $(\bar{\pi}, \ww) \in \calF_s(\calV)$;
and a vector $\tilde{\vv} = \bellman(s, \bar{\pi}, \ww) \in \neigh_\beta(\vv)$.

The size of the LP is bounded by $\poly(|S|, |\bfA|, L)$, given that the vertex representation of $\calV$ involves at most $L$ points. The time complexity then follows.
\end{proof}

\thmmainalg*

\begin{proof}
We analyze the time complexity of the two problems.

\paragraph{Computing a Fixed Point}
By \Cref{prp:alg-iteration-output}, \Cref{alg:value-set-iteration} terminates in $(1/\xi+2)^{n+1}$ iterations and outputs a desired $\calV$.
In each iteration, we compute $\xPhi(\calV)(s)$, which reduces to deciding, for each grid point $\vg \in G_\xi \cap \neigh_\xi(\calB)$, whether $\vg \in \neigh_\xi(\Phi(\calV)(s))$. The first part of \Cref{lmm:compute-inducibility} then applies, and the set of grid points inside $\neigh_\xi(\Phi(\calV)(s))$ gives a vertex representation of $\xPhi(\calV)(s)$.
There are $(1/\xi+2)^{n+1}$ grid points, so we solve the decision problem $(1/\xi+2)^{n+1}$ times, each taking time $\poly\left(|S|, |\bfA|, (1/\xi)^{n+1}\right)$ by \Cref{lmm:compute-inducibility}.
(The ``$+2$'' factor is absorbed by the $\poly(\cdot)$ operator.)
Overall, the time complexity of \Cref{alg:value-set-iteration} is bounded by $\poly\left(|S|, |\bfA|, (1/\xi)^{n+1}\right)$.

\paragraph{Computing $\pi(\sigma;s)$}
Given $\calV$ in vertex representation, such that
$\calV = \xPhi(\calV)$ and $\calV \supseteq \calV^\star$,
we first identify an optimal point in $\calV(s_\init)$ for the principal:
$\vv^* \in \argmax_{\vv \in \calV(s_\init)} v_0$.
This can be done by solving an LP now that we have the vertex representation of $\calV(s_\init)$.
The time it takes to solve the LP is bounded by a polynomial in the size of the LP; that is, $\poly\left((1/\xi)^{n+1}\right)$ given that the number of vertices in the vertex representation of $\calV(s_\init)$ is in the order of $(1/\xi)^{n+1}$.

Since $\calV \supseteq \calV^\star$, we have
$\tilde{v}_0 \ge \max_{\vv \in \calV^\star(s_\init)} v_0 - \epsilon$ for every $\tilde{\vv}$ in the $\xi$-neighborhood of $\vv^*$. 
So, by definition (\Cref{def:eps-delta-optimal-policy}), any $\delta$-CE inducing such $\tilde{\vv}$ is $(\epsilon,\delta)$-optimal.
Indeed, now that $\calV = \xPhi(\calV)$, we have $\Phi(\calV) \subseteq \calV \subseteq \neigh_\xi(\Phi(\calV))$ by \Cref{crl:approx-fixed}, which fulfills the requirement of \Cref{alg:compute-pi}.
So, by \Cref{lmm:eps-neigh-inducible-alg}, \Cref{alg:compute-pi} computes exactly such a $\delta$-CE.
It remains to analyze the time complexity of \Cref{alg:compute-pi}.

\paragraph{Time Complexity of \Cref{alg:compute-pi}}

By design, \Cref{alg:compute-pi} terminates in $|\sigma|$ iterations. 
The time it takes to find $\bar{\pi}$, $\ww$, and $\tilde{\vv}^t$ in each iteration follows from the second part of \Cref{lmm:compute-inducibility}. 
Specifically, we have 
\[
\vv^t \in \neigh_{\xi}(\Phi(\calV)(s^t))
\]
because $\vv^t \in \calV(s^t)$ (by \Cref{ln:vv-t}) while the algorithm requires $\calV \subseteq \neigh_\xi(\Phi(\calV))$.
So by \Cref{lmm:compute-inducibility}, we can find the desired $\bar{\pi}$, $\ww$, and $\tilde{\vv}^t$ in the stated amount of time.

A subtle representation issue requires some attention here: the bit size of $\tilde{\vv}^t$---which depends on $\tilde{\vv}^{t-1}$---may grow exponentially with $t$, in which case the time complexity stated in \Cref{lmm:compute-inducibility} will involve additional exponential terms.
This issue can be addressed by a rounding procedure, along with a finer grid with precision $\xi' = \xi/2$.

Specifically, replacing $\xi$ with $\xi'$ everywhere in our algorithms, we get that $\vv^t \in \neigh_{\xi'}(\Phi(\calV)(s^t))$ in each iteration.
In this case, there exists a hypercube in $\mathbb{R}^{n+1}$ that has side length $\xi'$ and contains both $\vv^t$ and a point in $\neigh_{\xi'}(\Phi(\calV)(s^t))$.
Any such hypercube must also contain a grid point $\vg \in G_{\xi'}$, which means that 
\[
\vg \in \neigh_{\xi'}(\vv^t) \cap \neigh_{\xi'}(\Phi(\calV)(s^t)).
\]
Rounding $\vv^t$ to this grid point ensures that its bit size does not grow with $t$.
Moreover: 
\begin{itemize}
\item 
$\vg \in \neigh_{\xi'}(\Phi(\calV)(s^t))$ ensures that we can find $(\bar{\pi}, \ww) \in \calF_{s^t}(\calV)$ such that $\tilde{\vv}^t \in \bellman(s^t, \bar{\pi}, \ww)$ for some $\tilde{\vv}^t \in \neigh_{\xi'}(\vg)$, which follows by a direct application of \Cref{lmm:compute-inducibility}. 

\item 
Given this, $\vg \in \neigh_{\xi'}(\vv^t)$ then ensures that
\[
\left\| \tilde{\vv}^t - \vv^t \right\|_\infty \;\le\;  
\left\| \tilde{\vv}^t - \vg \right\|_\infty + \left\| \vg - \vv^t \right\|_\infty \;\le\; 
2 \xi' \;=\; 
\xi.
\]
So $\tilde{\vv}^t \in \neigh_{\xi}(\vv^t)$, as desired at \Cref{ln:compute-vv} of \Cref{alg:compute-pi}.
\end{itemize}
The time complexity, after introducing this rounding procedure, remains $\poly\left(|S|, |\bfA|, (1/\xi)^{n+1},\, |\sigma| \right)$ since the additional coefficient $2$ of $1/\xi$ is absorbed by the $\poly(\cdot)$ operator.
\end{proof}

\subsection{Proofs in \Cref{sc:beyond-constant}}

\lmmtPhitermination*

\begin{proof}
The proof is similar to that of \Cref{prp:alg-iteration-output}.
The fact that $\calV= \tPhi(\calV)$ and $\calV \subseteq \calV^\star$ follows the same argument there. We bound the number of iterations next.

$\tPhi$ satisfies the same monotonicity property as that in \Cref{lmm:Phi-monotonicity}.
Moreover, it always maps $\neigh_{\xi/2}(\calB)$ to its subset based on the same argument in the proof of \Cref{prp:alg-iteration-output}.
Hence, by initializing the value sets to $\neigh_{\xi/2}(\calB)$, \Cref{alg:value-set-iteration} generates a sequence of value-set functions, each contained in the previous.
Each value-set function corresponds to the set of grid points it contains. 
By construction, if $\vg$ is contained in a value set $\tPhi(\calV)(\xx)$, then so are all the grid points $\vg' \in \neigh_{\xi/2}(\calB)$ such that $g_i = g'_i$ for all $i \in I_\xx$. This means that we can further simplify the representation by considering only the dimensions in $I_\xx$. 

Within the space spanning these dimensions, there are $(2/\xi+2)^{|I_\xx|} \le (2/\xi+2)^{(\lambda+1) \cdot c}$ grid points in total in $\neigh_{\xi/2}(\calB)$, where $I_\xx$ contains at most $c$ unique acting players in each of the $\lambda+1$ time steps involved.
Moreover, there are $|X| = |S|^{\lambda + 1}$ meta-states, which amounts to $|S|^{\lambda+1} \cdot (2/\xi+2)^{(\lambda+1) \cdot c}$ grid points in the simplified representation of a value-set function.
As the algorithm proceeds, at least one grid point must be removed from the representation in each iteration. The stated result then follows.
\end{proof}

\prpepsneighinducibleturnbased*

\begin{proof}
Similarly to the proof of \Cref{lmm:eps-neigh-inducible-alg}, we iteratively expand $\vv$ into onward vectors in the approach described in \Cref{alg:compute-pi}, and we show that this procedure defines a policy that induces $\vv$ approximately.

More specifically, let us first adapt \Cref{alg:compute-pi} to the meta-game: 
\begin{itemize}
\item We replace $S$ with the meta-state space $X$, and replace every state $s^t$ in the algorithm with a meta-state $\xx^t = (x_t^{-\lambda}, \dots, x_t^{-1}, x_t^{0})$.

\item We replace $\xi$ with $\xi' = \xi/2$, as in the definition of $\tPhi$.

\item We require $\calV = \tPhi(\calV)$
and use the new neighborhood notion, aiming to find some $\tilde{\vv}^t \in \neigh_{\xi'}^{I_{\xx^t}}(\vv^t)$ at $\xx^t$ at \Cref{ln:compute-vv}, where $\xx^t$ is the $t$-th meta-state in $\sigma$.
\end{itemize}

Indeed, we can prove that if $\calV = \tPhi(\calV)$, then $\calV$ satisfies $\Phi(\calV)(\xx) \subseteq \calV(\xx) \subseteq \neigh_{\xi'}^{I_\xx} \Big(\Phi(\vv)(\xx) \Big)$ for every $\xx \in X$ (similarly to \Cref{crl:approx-fixed}). In turn, \Cref{alg:compute-pi} produces a distribution for every input sequence $(\sigma;\xx)$, via a similar argument to the one used in the proof of \Cref{lmm:eps-neigh-inducible-alg}.
This gives a well-defined policy $\pi$.

It then remains to demonstrate that $\pi$ induces a value close to $\vv$.
We replicate the arguments in the proof of \Cref{lmm:eps-neigh-inducible-alg} where we compare the values in the original game as well as two variants of it: the subsidized game $\Gsub$ and truncated game $\Gtrun$.
By the same argument, $\pi$ induces exactly $\vv$ in both $\Gsub$ and $\Gtrun$.

To compare the values in $\Gsub$ and the original game, note that in the meta-game we cannot bound all the subsidies $\tilde{\rr}$ by $\xi'$ as in some dimensions, $\tilde{\vv}^t$ and $\vv^t$ may differ significantly due to the new neighborhood notion.
However, observe that by construction, for each iteration $t$, every player~$i \in I_{x_t^0}$---who acts in the current state $x_t^0$---will remain in $I_{\xx^\ell}$ for all $\ell = t+1, \dots, t+\lambda$. 
Within these time steps, we can bound the subsidies by $\xi'$ for every $i \in I_{x_t^0}$, while for the time steps outside of this range, we use the trivial bound $1+\xi$, which is the maximum $L_\infty$ distance between any two points in $\neigh_{\xi'}(\calB)$.
(It must be that $\calV \subseteq \neigh_{\xi'}(\calB)$, since otherwise $\calV = \tPhi(\calV)$ is not possible.)

This gives the following bound:
\begin{align*}
\| Q^{\pi,\rho}(\sigma;s, \va) - Q^{\pi,\rho, \tilde{\rr}}(\sigma; s, \va) \|_\infty 
&\le \sum_{\ell = t+1}^{t+\lambda} \gamma^{\ell-t} \cdot \xi' + \sum_{\ell = t+\lambda + 1}^{\infty} \gamma^{\ell-t} \cdot (1+\xi) \\
&\le \xi'/(1- \gamma) + (1+\xi) \cdot \gamma^\lambda / (1 - \gamma) \\
&< 2\xi'/(1- \gamma) \\
&\le \delta/2,
\end{align*}
where we used the fact that $(1+\xi) \cdot \gamma^\lambda < \xi'$ (\Cref{eq:lambda}).
Similarly, we can also bound the difference between the V-values by $\epsilon$. 
As a result, $\pi$ is a $\delta$-CE and it induces some $\tilde{\vv}$ at $s_\init = x_\init^0$ such that $|\tilde{v}_0 - v_0| \le \epsilon$.
This completes the proof.
\end{proof}

\thmmainalgcturn*

\begin{proof}
This can be proven by replicating the arguments in \Cref{thm:main-alg}.
Notably, the following upper bounds are used in the analysis:
the size $|X| = |S|^{\lambda+1}$ of the meta-state space; 
the size $|A|^{c+1}$ of the action space of the acting players (including $c$ agents and the principal); and
the number $(2/\xi+2)^{(\lambda+1)\cdot c}$ of grid points to be considered, which is due to the upper bound $(\lambda+1)\cdot c$ on $I_\xx$ for each $\xx \in X$.
In the stated time complexity, the term $2/\xi+2$ is further simplified to $1/\xi$ given the $\poly(\cdot)$ operator.
\end{proof}

\end{document}